\documentclass[11pt]{article}
\usepackage{odonnell}

\newcommand{\notecostin}[1]{}
\newcommand{\noteryan}[1]{}
\newcommand{\notejohn}[1]{}

\renewcommand\bra[1]{{\langle{#1}|}}
\makeatletter
\renewcommand\ket[1]{%
  \@ifnextchar\bra{\k@t{#1}\!}{\k@t{#1}}%
}
\newcommand\k@t[1]{{|{#1}\rangle}}
\makeatother

\newcommand{\Sym}[1]{\mathfrak{S}_{#1}}
\newcommand{\Srep}{\mathcal{P}}           % capital rho?
\DeclareMathOperator{\RepP}{\mathcal{P}}
\DeclareMathOperator{\RepQ}{\mathcal{Q}}
\newcommand{\on}{{\otimes n}}
\newcommand{\cyc}{\mathrm{cyc}}

\newcommand{\dtv}[2]{\mathrm{d}_{\mathrm{TV}}(#1,#2)}
\newcommand{\delltwo}[2]{\mathrm{d}_{\ell_2}(#1,#2)}
\newcommand{\delltwosq}[2]{\mathrm{d}^2_{\ell_2}(#1,#2)}
\newcommand{\dhell}[2]{\mathrm{d}_{\mathrm{H}}(#1,#2)}
\newcommand{\dhellsq}[2]{\mathrm{d}^2_{\mathrm{H}}(#1,#2)}
\newcommand{\BC}[2]{\mathrm{BC}(#1,#2)}
\DeclarePairedDelimiterX\diverg[2]{(}{)}{#1 \mathrel{}\mathclose{}\delimsize\|\mathopen{}\mathrel{} #2}
\newcommand{\dchisq}[2]{\mathrm{d}_{\chi^2}\diverg{#1}{#2}}
\newcommand{\Dtr}[2]{\mathrm{D}_{\mathrm{tr}}(#1,#2)}
\newcommand{\DHS}[2]{\mathrm{D}_{\mathrm{HS}}(#1,#2)}
\newcommand{\DHSsq}[2]{\mathrm{D}^2_{\mathrm{HS}}(#1,#2)}
\newcommand{\DB}[2]{\mathrm{D}_{\mathrm{B}}(#1,#2)}
\newcommand{\DBsq}[2]{\mathrm{D}^2_{\mathrm{B}}(#1,#2)}
\newcommand{\Fid}[2]{\mathrm{F}(#1,#2)}
\newcommand{\Fidsq}[2]{\mathrm{F}^2(#1,#2)}
\newcommand{\DBchi}[2]{\mathrm{D}_{\chi^2}\diverg{#1}{#2}}

\newcommand{\chiobs}{\mathcal{X}_\sigma}
\newcommand{\chiobsij}[2]{\chiobs^{(#1,#2)}}
\newcommand{\chiobsavg}{\mathcal{O}_{\chi^2}}
\newcommand{\trtwo}[2]{\omega^{(2)}_\sigma(#1,#2)}
\newcommand{\trthree}[3]{\omega^{(3)}_\sigma(#1,#2,#3)}

\newcommand*{\cO}{\mathcal{O}}

\newcommand*{\gS}{\mathfrak{S}}
\newcommand*{\tensor}{\otimes}
\newcommand*{\placeholder}{\,\cdot\,}

\DeclareMathOperator{\Ad}{Ad}
\DeclareMathOperator{\End}{End}
\DeclareMathOperator{\GL}{GL}
\DeclareMathOperator{\U}{U}

\newcommand{\purityestimator}[1]{\mathsf{TN}(#1)}

\newcommand{\Specht}[1]{\mathrm{Sp}_{#1}}
\newcommand{\Weyl}[2]{\mathrm{V}_{#1}^{#2}}

\newcommand{\sirrep}[1]{p_{#1}}

\newcommand{\types}[2]{\mathrm{Types}^{#1}_{#2}}

\newcommand{\SYT}[1]{\mathrm{SYT}_{#1}}
\newcommand{\SSYT}[2]{\mathrm{SSYT}_{#1}^{#2}}

\newcommand{\bmu}{\boldsymbol{\mu}}
\newcommand{\bnu}{\boldsymbol{\nu}}
\newcommand{\btau}{\boldsymbol{\tau}}

\begin{document}

\title{Quantum state certification}
\author{Costin B\u{a}descu\thanks{Computer Science Department, Carnegie Mellon University.  Supported by NSF grant CCF-1618679. \texttt{\{cbadescu,odonnell\}@cs.cmu.edu}} \and Ryan O'Donnell${}^*$ \and John Wright\thanks{Center for Theoretical Physics, Massachusetts Institute of Technology. Supported by NSF grant CCF-6931885. \texttt{jswright@mit.edu}}}

\maketitle

\begin{abstract}
    We consider the problem of \emph{quantum state certification}, where one is given $n$ copies of an unknown $d$-dimensional quantum mixed state~$\rho$, and one wants to test whether $\rho$ is equal to some known mixed state~$\sigma$ or else is $\eps$-far from~$\sigma$.  The goal is to use notably fewer copies than the $\Omega(d^2)$ needed for full tomography on~$\rho$ (i.e., density estimation).  We give two robust state certification algorithms: one with respect to fidelity using $n = O(d/\eps)$ copies, and one with respect to trace distance using $n = O(d/\eps^2)$ copies. The latter algorithm also applies when $\sigma$ is unknown as well.  These copy complexities are optimal up to constant factors.
\end{abstract}

\section{Introduction}

A key step in building quantum devices is verifying that they work as intended. Typically, a quantum device is designed with the intent of outputting some known $d$-dimensional (mixed) state $\sigma \in \C^{d \times d}$, but the possibility of imperfections in the device's construction and noise in the device's operation mean that its actual output state $\rho \in \C^{d \times d}$ is unknown.  \emph{Quantum state certification} refers to the problem of testing whether~$\rho$ equals~$\sigma$ or is far from~$\sigma$, given the ability to produce~$\rho^{\otimes n}$ (i.e.,~$n$ copies of~$\rho$).  This is the quantum (noncommutative) generalization of the classical statistical problem of testing identity of probability distributions~\cite{Can15}.

A standard approach for quantum state certification is to first estimate~$\rho$ from $\rho^{\on}$ using a \emph{quantum state tomography (estimation)} procedure, then to check that the estimate is close to~$\sigma$.
Given that~$\rho$ has $d^2-1$ real parameters, it is natural that the number of copies needed to estimate it should scale roughly as~$d^2$.  This was confirmed in a trio of recent papers~\cite{HHJ+16,OW16,OW17}; among other things, those works show that $n = \wt{\Theta}(d^2/\eps)$ copies of~$\rho$ are necessary and sufficient to produce an estimate~$\wh{\rho}$ satisfying the fidelity bound $\Fid{\rho}{\wh{\rho}} \geq 1-\eps$.  (See \Cref{sec:prior-quantum} for more on prior work, and \Cref{sec:distances} for a review of distance measures such as fidelity, trace distance, $\chi^2$-divergence, etc.)

Unfortunately, even small scale quantum systems can have large dimension; for example, a system of~$q$ qubits has $d = 2^q$ dimensions.  For such systems, the quadratic scaling in~$d$ required by full tomography (density estimation) can be prohibitively expensive.  For example, a 2005 experiment~\cite{HHR+05} designed to produce the entangled $8$-particle $W$-state ($d = 256$) used
%Real world examples of tomography have involved using $n=605$ copies to learn a single qubit ($d=2$)~\cite{MHS+12} and
$n=656100$ copies to estimate the actually-produced state.  (The fidelity to the target state ended up being estimated as $.85$.)

However for the quantum state certification problem, the goal is not to learn the unknown state~$\rho \in \C^{d \times d}$ but merely to test whether it is close to a target $\sigma$, or far from it.  Learning the entire density matrix might be wasting copies of~$\rho$ to gain irrelevant information. As such, it is natural to ask: can we outperform tomography?

\subsection{Our results}
In this work, we give a unified framework for analyzing the number of copies of~$\rho$ needed to estimate polynomial functions of~$\rho$ and hence perform various quantum state certification tasks.    One of our main results is the following:
\begin{theorem}                                     \label{thm:maybe-main}
    Let $\sigma \in \C^{d \times d}$ be a fixed mixed state, and let $\eps > 0$.  There is an algorithm that, given $n = O(d/\eps)$ copies of $\rho$, performs a measurement and then reports either ``close'' or ``far''.  The algorithm has the following guarantee (with high probability\footnote{Henceforth abbreviated ``whp''.  We may take this to mean probability at least, say,~$2/3$; however, by standard means this probability can be boosted to $1-\delta$ at the expense of multiplying $n$ by $O(\log(1/\delta))$.}):  If it reports ``close'' then we have the fidelity bound $\Fid{\rho}{\sigma} \geq 1-\eps$. If it reports ``far'' then we have the Bures $\chi^2$-divergence\footnote{The Bures $\chi^2$-divergence is reviewed in \Cref{sec:distances}.} bound $\DBchi{\rho}{\sigma} > .49\eps$.
\end{theorem}
To put it another way, if $\DBchi{\rho}{\sigma} \leq .49\eps$ (in particular, if $\rho = \sigma$) then the algorithm reports ``close'' and if  $\Fid{\rho}{\sigma} < 1-\eps$ then the algorithm reports ``far'' (whp).  We remark that the notions of ``close'' and ``far'' in \Cref{thm:maybe-main} are nearly complementary, since it's known that every pair of states $\rho, \sigma$ satisfies either $\Fid{\rho}{\sigma} \geq 1-\eps$ or $\DBchi{\rho}{\sigma} > .5\eps$.

\Cref{thm:maybe-main} is stronger than the usual kind of state certification result in that it is \emph{robust}, meaning that the test ``accepts'' not just if $\rho = \sigma$ but also if $\rho$ is sufficiently close to~$\sigma$.  The simplified (weaker) version would be:
\begin{corollary}                                       \label{cor:main}
    For a fixed mixed state $\sigma \in \C^{d \times d}$  and $\eps > 0$, there is an algorithm that, given $n = O(d/\eps)$ copies of $\rho$, distinguishes (whp) between the cases $\rho = \sigma$ and $\Fid{\rho}{\sigma} < 1-\eps$.
\end{corollary}

The stronger version \Cref{thm:maybe-main} is actually an easy consequence (see \Cref{sec:chi2-consequences}) of the following certification procedure for ``well-conditioned'' states, robust with respect to Bures $\chi^2$-divergence:
\begin{theorem}\label{thm:robust-chi-sq}
    Let $c  > 0$ be any small constant. Fix a $d$-dimensional mixed state~$\sigma$ with smallest eigenvalue at least $c \epsilon^2/d$. Then there is an algorithm that, given $n = O(d/\eps^2)$ copies of~$\rho$, (whp) outputs ``close'' if $\DBchi{\rho}{\sigma} \leq .99 \eps^2$ and outputs ``far'' if $\DBchi{\rho}{\sigma} > \eps^2$.
\end{theorem}

We also obtain a new sample-efficient certification algorithm in the case of \emph{two unknown} states.  Here one is given $n$ copies each of mixed states $\rho, \sigma$ and one wants to distinguish whether $\rho = \sigma$ or $\rho$ is far from~$\sigma$.  Our algorithm here is robust with respect to the Hilbert--Schmidt distance:
\begin{theorem}\label{thm:robust-HS}
    There is an algorithm that, given $n = O(1/\eps^2)$ copies each of unknown mixed states $\rho, \sigma \in \C^{d \times d}$, (whp) outputs ``close'' if $\DHS{\rho}{\sigma} \leq .99 \eps$ and outputs ``far'' if $\DHS{\rho}{\sigma} > \eps$.
\end{theorem}
Of course this result may also be used in the simpler case when~$\sigma$ is a known state (as then the algorithm can simply prepare~$n$ copies of $\sigma$ by itself).  We also remark that the sample complexity~$n$ has no dependence on~$d$.

Although the Hilbert--Schmidt distance is arguably not too meaningful, operationally, one can use Cauchy--Schwarz to relate it to the very natural trace distance.  In this way, \Cref{thm:robust-HS} immediately yields the following:
\begin{corollary}                                       \label{cor:basic-trace}
    There is an algorithm that, given $n = O(d/\eps^2)$ copies each of unknown mixed states $\rho, \sigma \in \C^{d \times d}$, (whp) distinguishes between the cases  $\rho = \sigma$ and $\Dtr{\rho}{\sigma} > \eps$.
\end{corollary}
We stated the above corollary for simplicity, but with slightly more care (see \Cref{sec:test-conseq1}) one also derive from \Cref{thm:robust-HS} the following much more precise result for trace-distance certification, which has improved sample complexity when one of the states is close to having low rank:
\begin{corollary}                                       \label{cor:pca-trace}
    Assume that one of the two unknown states --- say, $\sigma$ --- is close to having rank at most~$k$, in the sense that the sum of its largest~$k$ eigenvalues is at least $1-\delta$.  Then there is an algorithm that, given $n = O(k/\eps^2)$ copies each  of $\rho, \sigma \in \C^{d \times d}$, (whp) distinguishes between the cases $\DHS{\rho}{\sigma} \leq .58 \eps$ and $\Dtr{\rho}{\sigma} > \delta + \eps$.  (The constant $.58$ can be anything smaller than $2-\sqrt{2}$.)
\end{corollary}

We note that even the simplest versions of our results --- \Cref{cor:main} and \Cref{cor:basic-trace} --- have optimal sample complexity (up to a constant), even when $\sigma$ is promised to be the maximally mixed state~$\Id/d$. This is a consequence of the following lower bound from~\cite{OW15}:
\begin{theorem}[\cite{OW15}]                                     \label{thm:paninski-lower}
    Given even $d$ and $0 \leq \eps \leq 1/2$, let $\sigma = \Id/d$ and let $\calC_\eps$ denote the class of states with eigenvalues $\frac{1+2\eps}{d}, \frac{1-2\eps}{d}, \frac{1+2\eps}{d}, \frac{1-2\eps}{d}, \dots, \frac{1+2\eps}{d}, \frac{1-2\eps}{d}$.  For any $\rho \in \calC_\eps$, one has 
    \[
        \Dtr{\rho}{\sigma} = \eps, \quad \Fid{\rho}{\sigma} = 1 - \tfrac12 \epsilon^2 - O(\eps^4), \quad \DHS{\rho}{\sigma} = 2\eps/\sqrt{d}, \quad \DBchi{\rho}{\sigma} = 4\eps^2 + O(\eps^4).
    \]
    Then any measurement strategy that can distinguish (with probability advantage at least~$1/3$) the case $\rho = \sigma$ from the case $\rho \in \calC_\eps$ using $n$ samples from~$\rho$ must have $n > .15d/\eps^2$.
\end{theorem}

Finally, our quantum certification algorithm from \Cref{thm:robust-HS} is not just copy-efficient, it can be carried out by polynomial-sized (i.e., $\poly(n, d)$-gate) quantum circuits.

\subsection{Outline of the remainder of the paper}
In \Cref{sec:prior} we review prior work on quantum tomography and state discrimination, as well as some relevant prior work on classical learning and testing of probability distributions.   In \Cref{sec:distances} we recall various measures of probability distribution distance and quantum state distance that will be important in this work.  \Cref{sec:quantum-probability,sec:rep} are devoted to background on quantum probability and representation theory.   In \Cref{sec:efficiency}, we develop a framework for finding the most efficient (lowest-variance) estimators for symmetric polynomial functions of unknown quantum states.  These results are not strictly necessary for our proof of \Cref{thm:robust-HS} in \Cref{sec:hs}; however, they justify that the estimators used therein are optimal.  \Cref{sec:chisq} contains our proof of \Cref{thm:maybe-main,thm:robust-chi-sq}, as well as a diagonality tester for quantum states.
Finally, in \Cref{sec:implementing} we give efficient implementations for the algorithm in \Cref{thm:robust-HS}.

\section{Prior work on classical and quantum density testing/estimation} \label{sec:prior}

In this section we review some results on learning and testing unknown quantum states, and the analogous classical problem of learning and testing unknown probability distributions.  As these areas are extremely broad, we cannot completely review all known literature; we will simply give pointers to some of the best known and most relevant results.

\subsection{Prior quantum density estimation, testing, and certification}  \label{sec:prior-quantum}

\subsubsection{Tomography (density estimation)}
Before discussing state certification, we start by reviewing the best known results for the baseline problem of \emph{tomography}; i.e., producing an estimate $\wh{\rho}$ of an unknown density matrix $\rho \in \C^{d \times d}$, given $n$ copies $\rho^{\otimes n}$, up to error~$\eps$ (whp) for some notion of ``distance''. We will also let $k$ denote the \emph{rank} of $\rho$, which is $1$ when $\rho$ is a pure state, and in general is at most~$d$.  The best results achievable depend on the ``figure of merit'' --- i.e., distance measure --- chosen (see \Cref{sec:distances} for a review).

In~\cite{HHJ+16} it was shown that  $n = O(kd/\epsilon) \cdot \log(d/\epsilon)$ copies suffice to obtain infidelity~$\eps$ (i.e., $\Fid{\rho}{\wh{\rho}} \leq 1 - \eps$); this also implies that $n = O(kd/\epsilon^2) \cdot \log(d/\epsilon)$ copies suffice to obtain trace distance~$\eps$ (i.e., $\Dtr{\rho}{\wh{\rho}} \leq \eps$).  Those authors also showed that $n = \Omega(kd/\epsilon^2) / \log(d/k\epsilon)$ copies are necessary, with the $\log$ factor being removable in the case $k = d$.  Independently, in~\cite{OW16} it was shown that $n = O(d/\eps^2)$ copies suffice to obtain Hilbert--Schmidt distance~$\eps$ (i.e., $\DHS{\rho}{\wh{\rho}} \leq \eps$); this also implies a copy complexity of $n = O(kd/\epsilon^2)$ for trace distance (slightly better than in~\cite{HHJ+16}).  More generally, \cite{OW16} showed a kind of ``PCA'' result: for $\rho$ of any rank, $n = O(kd/\epsilon^2)$ copies suffice to produce an estimate $\wh{\rho}$ whose trace distance from~$\rho$ is at most $\eps$ more than that of the best rank-$k$ approximator.  Finally, a followup work~\cite{OW17} gave an alternate proof of the $n = O(kd/\epsilon) \cdot \log(d/\epsilon)$ bound for infidelity, showed also an $n = O(k^2d/\epsilon)$ bound, and extended these bounds to the PCA case.

\subsubsection{Density testing}

Tomography results suffer from the inherent issue that $n = \wt{\Theta}(d^2)$ copies are needed in the general case (except when the figure of merit is Hilbert--Schmidt distance, but this metric is not considered to be very meaningful, operationally).  Thus as mentioned, it is natural to focus on restricted problems like state certification, distance estimation, and other \emph{property testing} problems that can potentially be carried out with $n = O(d)$ or better.  Montanaro and de~Wolf~\cite{MdW16} have given an excellent survey on property testing of quantum states; we review a few of the known results here.

A typical quantum property testing problem would involve two disjoint classes $\calC_1$, $\calC_2$ of $d$-dimensional quantum states; given $n$ copies of an unknown $\rho$, promised to be in either~$\calC_1$ or~$\calC_2$, the task is to distinguish which is the case (whp) using few copies of~$\rho$.  In particular, the \emph{quantum state certification} problem for fixed state $\sigma \in \C^{d \times d}$ is the case when $\calC_1 = \{\sigma\}$ and $\calC_2 = \{ \rho : \mathrm{D}(\rho,\sigma) > \eps\}$ for some notion $\mathrm{D}(\cdot, \cdot)$ of distance and some parameter~$\eps$.

When $\sigma$ is a \emph{pure} state, it is straightforward to show (see, e.g.,~\cite{MdW16}) that the associated quantum state certification task, with infidelity as the distance measure, can be done using $n = O(1/\eps)$ copies (and this implies $n = O(1/\eps^2)$ copies suffice for trace distance).  Indeed, the same is possible when both $\rho$ \emph{and} $\sigma$ are unknown pure states, and one is given $n$ copies of each.  For practical purposes, it may be useful to have a state certification algorithm for a known pure~$\sigma$ that only uses simple measurements; e.g., Pauli observables.  For this problem, it has been shown~\cite{FL11,SLP11,AGKE15} that for $\sigma$ known and pure, one can solve the certification problem given $n = O(d/\eps^2)$ copies of an unknown~$\rho$ with infidelity as the distance metric --- indeed, with this many copies one can estimate the fidelity $\Fid{\rho}{\sigma}$ to $\pm \eps$.

For the state certification problem when $\sigma$ is mixed (not pure), not much is known except in one case: when $\sigma = \Id/d$, the ``maximally mixed'' state.  For this problem, it was shown in~\cite{OW15} that $n = \Theta(d/\eps^2)$ copies are necessary and sufficient, when the distance measure is trace distance.  In fact, for the $n = O(d/\eps^2)$ upper bound, \cite{OW15}~effectively show that one can estimate the \emph{purity} $\tr(\rho^2)$ of~$\rho$ sufficiently well so as to distinguish between purity $1/d$ (achieved by the maximally mixed state) and purity exceeding $1/d + \eps^2/d$.  Note that the latter case is equivalent to $\rho$ being $\epsilon/\sqrt{d}$-far from $\Id/d$ in Hilbert--Schmidt distance and $\epsilon^2$-far from $\Id/d$ in Bures $\chi^2$-divergence. The lower bound was mentioned earlier as \Cref{thm:paninski-lower}.

\subsubsection{The asymptotic regime for state discrimination}  \label{sec:asympt-state-discrim}

There is a related class of work that we refer to as the ``asymptotic regime''.  Consider the simplest quantum property testing problem, \emph{state discrimination}, in which $\calC_1 = \{\sigma_1\}$ and $\calC_2 = \{\sigma_2\}$ for two known states $\sigma_1, \sigma_2 \in \C^{d \times d}$.  The perspective we take in this paper involves determining the least number of copies~$n$ such that one can distinguish $\rho = \sigma_1$ from $\rho = \sigma_2$ with high probability --- say, with both ``type I'' and ``type II'' errors having probability at most~$\delta = 1/3$.  One can reduce this~$\delta$ to any small positive constant at the expense of making~$n$ a constant factor larger.  We refer to this perspective as the \emph{non-asymptotic regime}, because we do not consider any limiting error rate as $n \to \infty$; rather, we  wish to find a concrete upper bound on the~$n$ that suffices, depending only on $d$, the distance between $\sigma_1$ and $\sigma_2$, and nothing else.

On the other hand, there is substantial work on the \emph{asymptotic regime}, sometimes going under the name \emph{quantum hypothesis testing}, in which the focus is on how exponentially fast the error rate goes to~$0$ in the limit as $n \to \infty$.  Here one might seek the best (smallest) constant~$R$ such that, given $n$ copies, one can ensure type~I and type~II errors have probability at most $(R+o(1))^n$, where the $o(1)$ refers to $n  \to \infty$.  A downside of such results is that they do not a priori give any information about how large $n$ needs to be before error bounds ``kick in''; e.g., the $o(1)$ function might not be less than, say,~$.1$ until $n$ is larger than some uncontrolled function of~$d$ (e.g., $2^d$) or of some other parameters (e.g., the smallest nonzero eigenvalue of~$\sigma_1$ or $\sigma_2$).

A good survey of the results in the asymptotic regime appears in~\cite{ANSV08}; they review known quantum versions of Stein's Lemma and Sanov's Theorem, and  prove quantum versions of Chernoff's Bound and the the Hoeffding--Blahut--Csisz\'{a}r--Longo bound.  For example, in the basic hypothesis testing problem described above, they prove that the best rate~$R$ is given by $Q_{\text{min}}\diverg{\sigma_1}{\sigma_2} = \min_{0 \leq s \leq 1}\tr(\sigma_1^s \sigma_2^{1-s})$ (a quantity that is within a factor of~$2$ of the infidelity between $\sigma_1$ and~$\sigma_2$).

\subsection{Prior classical density estimation, testing, and certification}

For every quantum problem discussed so far, we get a ``classical'' special case by assuming that all $d$-dimensional density matrices are diagonal.  In this way we obtain basic problems in statistics and property testing: estimation, certification, and identity testing for \emph{probability distributions} $p = (p_1, \dots, p_d)$ on~$[d]$.  Since our results are partly inspired by these classical analogues, we briefly review some known results here.

\subsubsection{Density estimation}
The analogue of quantum tomography is \emph{density estimation}: producing an estimate $\wh{p}$ of an unknown probability distribution $p$ on $[d]$, given $n$ independent samples.  For this problem, the most natural algorithm is simply to let $\wh{p}$ be the empirical distribution of the samples.  One can very easily directly calculate that
\[
    \E[\delltwosq{p}{\wh{p}}] = \frac1n\parens*{1- \littlesum_{i=1}^d p_i^2} \leq \frac1n;
\]
hence Markov's inequality implies that $n = O(1/\eps^2)$ samples suffice to obtain $\delltwosq{p}{\wh{p}} \leq \eps$ whp.  Cauchy--Schwarz then implies that $n = O(d/\eps^2)$ samples suffice to obtain $\dtv{p}{\wh{p}} \leq \eps$ with high probability.  For the stronger $\chi^2$-divergence, one shouldn't let $\wh{p}$ be the empirical distribution because then $\dchisq{p}{\wh{p}} = \infty$ is possible if $p_i  > \wh{p}_i = 0$ for some~$i$.  Instead, standard practice is to take $\wh{p}$ to be the ``add-one'' estimator: $\wh{p}_i = \frac{\bx_i + 1}{n+d}$, where $\bx_i \sim \text{Bin}(n,p_i)$ is the number of~$i$'s in the sample. Again, one can very easily directly calculate (see, e.g.,~\cite[Lemma~4]{KOPS15}):
\begin{equation}    \label{eqn:classical-chisq}
    \E[\dchisq{p}{\wh{p}}] = \frac{d-1}{n+1} - \frac{n+d}{n+1}\parens*{1 - \littlesum_{i=1}^d(1-p_i)^{n+1}}\leq \frac{d-1}{n+1} \leq \frac{d}{n}
\end{equation}
and hence $n = O(d/\eps)$ samples suffice to obtain $\dchisq{p}{\wh{p}} \leq \eps$ whp. Thus for natural measures of discrimination like total variation, Hellinger, and $\chi^2$-divergence, $n = O(d)$ samples suffice for density estimation (for constant~$\eps)$.  Consequently, for the ``distribution certification'' problem (known in property testing problems as ``identity testing''), the goal is to use $o(d)$ samples.

\subsubsection{Identity testing}
Three of the main such property testing problems, in increasing order of difficulty, are the following:
    \begin{enumerate}
        \item[0.] Testing identity of~$p$ to the uniform distribution (which we write as $\Id/d$ in this section).
        \item Testing identity of~$p$ to an arbitrary but known distribution~$q$.
        \item Testing identity of two unknown distributions~$p,q$.
    \end{enumerate}
\paragraph{Uniformity testing.} Historically, property testing researchers considered total variation distance to be the main figure of merit.  But beginning with the earliest work of Goldreich and Ron~\cite{GR00}, it was found that approaching the problems via $\ell_2^2$-distance was more expedient.    For example, Goldreich and Ron originally showed that with $n = O(\sqrt{d}/\eps^2)$ samples, one can (whp) estimate $\|p\|_2^2 = \delltwosq{p}{\Id/d} + 1/d$ to a multiplicative $1 \pm \eps$ factor.  A consequence of this (and Cauchy--Schwarz) is that $n = O(\sqrt{d}/\eps^4)$ samples suffice to distinguish $p = \Id/d$ and $\dtv{p}{\Id/d} > \eps$.  Paninski~\cite{Pan08} improved the latter result by a different method to $n = O(\sqrt{d}/\eps^2)$ (assuming $\eps = \Omega(d^{-1/4})$, a restriction later removed in~\cite{VV17}), and showed a matching lower bound.

In fact, a better analysis of Goldreich and Ron's original method yields the optimal result: one simply estimates $\delltwosq{p}{\Id/d}$ by the natural unbiased estimator (the average number of ``collisions'' among the~$n$ samples, minus $1/d$), computes its variance, and then uses Chebyshev inequality.  A little case analysis is needed when applying Chebyshev, which is perhaps why this natural method was not employed until the very recent work of~\cite{DGPP16} (for a briefer exposition, see~\cite[Sec.~10]{OW17}).  We will use similar methods in the present work, and the needed version of Chebyshev's inequality is packaged up at the end of this section as \Cref{lem:chebyshev}.

\paragraph{Identity testing to a known distribution.} Moving on to Problem~1 above, testing identity of $p$ to an arbitrary known distribution~$q$, Batu et~al.~\cite{BFF+01} showed that $O(\sqrt{d} \log(d) /\eps^2)$ samples from~$p$ suffice to distinguish the case $\dtv{p}{q} \ll \eps^3/\sqrt{d}\log(d)$ (and in particular, $p = q$) from the case $\dtv{p}{q} > \eps$. Valiant and Valiant~\cite{VV17} removed the $\log d$ factor from the sample complexity (though without analyzing ``robustness'').  The analysis in these works showed the importance at looking at ``weighted'' versions of the $\ell_2^2$-distance $\delltwosq{p}{q} = \sum_i (p_i-q_i)^2$ in which the $i$th summand is reweighted by a factor depending on~$q_i$.  Indeed, Acharya~et~al.~\cite{ADK15} improved these results by considering an unbiased estimator for the $\chi^2$-divergence of~$p$ from~$q$ and (implicitly) using a form of \Cref{lem:chebyshev}; they showed that $n = O(\sqrt{d}/\eps^2)$ samples from~$p$ suffice to distinguish $\dchisq{p}{q} \leq \eps^2/10$ from $\dtv{p}{q} > \eps$.  Indeed, although it is not stated this way, a close inspection of their proof shows that they actually obtain a robust tester for $\chi^2$-divergence under the assumption that $q_i \geq \Omega(\eps^2/d)$ for all~$i$.  This observation motivated our result \Cref{thm:robust-chi-sq}.  As a not too difficult consequence, the present authors and others~\cite{DKW17} observed that one can upgrade the~\cite{ADK15} result from ``$\chi^2$-vs.-$\ell_1$'' to the strictly superior ``$\chi^2$-vs.-Hellinger'', \`{a}~la~our \Cref{thm:maybe-main}.

On the subject of testing identity of $p$ to a known distribution~$q$, we should mention the line of work on ``instance-optimal'' results due to Valiant and Valiant~\cite{VV17} and Blais et~al.~\cite{BCG17}.  Stating these is slightly technical, but roughly speaking they show that one can distinguish $p = q$ from $\dtv{p}{q} > \eps$ using just $n = O(\sqrt{k}/\eps^2)$ samples provided the largest $k$ values of~$q$ sum to at least $1 - \Theta(\eps)$.  This can be compared with our \Cref{cor:pca-trace}.

\paragraph{Identity testing with two unknown distributions.} Finally, we discussed Problem~2 mentioned above, testing identity of \emph{two} unknown distributions $p$ and $q$ on $[d]$, given $n$ samples from each.  This problem was first studied by Batu et~al.~\cite{BFR+13}, who used a natural estimator for $\delltwosq{p}{q}$ to show that $n = O(1/\eps^4)$ samples suffice to distinguish $\delltwosq{p}{q} \leq \eps/2$ from $\delltwosq{p}{q} > \eps$.  (This has no dependence on~$d$ but a nonoptimal dependence on~$\eps$; in fact, our \Cref{thm:robust-HS} improves on this, even in the quantum case.)  From this, they were able to derive a total variation tester, using $n = O(d^{2/3} \log(d)/\eps^{8/3})$ samples to distinguish $p = q$ from $\dtv{p}{q} > \eps$ (in fact, they had a robust condition in place of~$p = q$). This was improved by Chan et~al.~\cite{CDVV14} to an optimal bound of $n = O(\max\{d^{2/3} / \eps^{4/3}, \sqrt{d}/\eps^2\})$ by means of an estimator resembling the Le Cam (triangular) discrimination.  The result was later reproved by Diakonikolas and Kane~\cite{DK16}, who also obtained a tester for Hellinger distance in the case of unknown $p$ and~$q$ with near-optimal sample complexity of $n = \wt{O}(\min \{d^{2/3}/\eps^{8/3}, d^{3/4}/\eps^2\})$ (improving on an $n = \wt{O}(d^{2/3} /\eps^8)$ bound of Guha et~al.~\cite{GMV09}).
Subsequently, the tilde on the big-Oh was removed by~\cite{DKW17}, giving the optimal sample complexity for this case.
We remark that obtaining an analogous result in the quantum case is an interesting open problem (specifically, obtaining an identity testing algorithm for two unknown states $\rho, \sigma$ that uses $n = O(d/\eps)$ samples to distinguish $\rho = \sigma$ from $\Fid{\rho}{\sigma} < 1-\eps$).

\bigskip

We end this section by stating and proving the useful version of Chebyshev described earlier.
\begin{lemma}                                       \label{lem:chebyshev}
    Let $\bX^{(n)}$ be a sequence of estimators for a number $\mu \geq 0$, meaning $\E[\bX^{(n)}] = \mu$ for all~$n$.  Suppose we have a variance bound of the form
    \begin{equation}    \label{eqn:variance}
        \Var[\bX^{(n)}] \leq O\parens*{\frac{b(\mu)}{n^2} + \frac{v(\mu)}{n}},
    \end{equation}
    where
    \begin{equation}    \label{eqn:increasing}
        b(\mu), \ v(\mu), \ \frac{\mu^2}{b(\mu)}, \ \frac{\mu^2}{v(\mu)} \ \text{are increasing functions of~$\mu \geq 0$.}
    \end{equation}
    (The $O(\cdot)$ should hide a universal constant.)  Let $\theta > 0$ be a parameter.  Then provided
    \begin{equation}    \label{eqn:how-many}
        n \geq C \max\braces*{\sqrt{\frac{b(\theta)}{\theta^2}}, \frac{v(\theta)}{\theta^2}},
    \end{equation}
    one can use $\bX^{(n)}$ to distinguish (with high probability) whether $\mu \leq .99 \theta$ or $\mu > \theta$.  Here $C$ is another universal constant. (More generally, to achieve $1-\gamma$ in place of~$.99$, one should take $\gamma^2 \theta^2$ in place of $\theta^2$ in the denominators in \eqref{eqn:how-many}.)
\end{lemma}
\begin{proof}
    We report ``$\mu \leq .99 \theta$'' if $\bX^{(n)} \leq .995 \theta$ and report ``$\mu > \theta$'' if $\bX^{(n)} > .995 \theta$.

    To analyze the correctness, suppose first that $\mu \leq .99 \theta$.  Then $b(\mu) \leq b(.99 \theta) \leq b(\theta)$ (by~\eqref{eqn:increasing}) and similarly $v(\mu) \leq v(\theta)$. Using these inequalities in~\eqref{eqn:variance} and then substituting in~\eqref{eqn:how-many}, we get $\Var[\bX^{(n)}] \leq O(\frac{1}{C^2} + \frac{1}{C}) \theta^2$.  For $C$ sufficiently large this implies $\stddev[\bX^{(n)}] \leq .001 \theta$, say, and then Chebyshev implies $\bX^{(n)} \leq \mu + .005 \theta \leq .995\theta$ with high probability.

    On the other hand, suppose that $\mu > \theta$.  Then $\frac{b(\theta)}{\theta^2} \geq \frac{b(\mu)}{\mu^2}$ (by~\eqref{eqn:increasing}) and similarly $\frac{v(\theta)}{\theta^2} \geq \frac{v(\mu)}{\mu^2}$.  Using these inequalities in~\eqref{eqn:how-many} and then substituting into~\eqref{eqn:variance}, we get $\Var[\bX^{(n)}] \leq O(\frac{1}{C^2} + \frac{1}{C}) \mu^2$.  For~$C$ sufficiently large this implies $\stddev[\bX^{(n)}] \leq .001\mu$, say, and then Chebyshev implies $\bX^{(n)} \geq \mu - .005\mu > .995\theta$ with high probability.
\end{proof}

\section{Preliminaries}

\subsection{Classical and quantum distances and divergences} \label{sec:distances}

\subsubsection{Distances and divergences for classical probability distributions}
There are many distances and divergences used for comparing discrete probability distributions $p = (p_1, \dots, p_d)$ and $q = (q_1, \dots, q_d)$; see, e.g.,~\cite{GS02,Cro17}.  We review some important ones here.  All of the distances we review will be \emph{permutation invariant}, meaning they are unchanged if the same permutation $\pi \in \Sym{d}$ is simultaneously applied to the outcomes of $p$ and $q$.
\begin{definition}
    The \emph{total variation distance} between $p$ and $q$ is
    \[
        \dtv{p}{q} = \frac12 \sum_{i=1}^d \abs*{p_i - q_i} = \frac12 \|p-q\|_1.
    \]
\end{definition}
\noindent The total variation distance is a metric and has a maximum value of~$1$, occurring when $p$ and $q$ have disjoint support.  It also has an operational meaning: it is the greatest probability with which one can discriminate a draw from $p$ and a draw from $q$; i.e., $\dtv{p}{q} = \max_{A \subseteq [d]}\{\abs{\Pr_{p}[A] - \Pr_q[A]}\}$.
\begin{definition}
    The \emph{$\ell_2$ distance} between $p$ and $q$ is
    \[
        \delltwo{p}{q} = \Bigl(\sum_{i=1}^d \parens*{p_i-q_i}^2\Bigr)^{1/2} = \|p-q\|_2.
    \]
\end{definition}
\noindent %In the context of probability distributions, some authors include a factor of~$\frac12$ so that the maximum~$\ell_2$ distance is~$1$; we will not do so.
The $\ell_2$ distance is also a metric; nevertheless we more often consider its square, $\delltwosq{p}{q}$. As a probability metric, the $\ell_2$ distance is somewhat unnatural.  For example, it does not satisfy the ``data processing inequality'', meaning that there is a stochastic operation that \emph{increases} $\ell_2$ distance.  However it is by far the easiest distance to calculate, as $\delltwosq{p}{q}$ is a simple polynomial in~$p$ and~$q$; further, it can be related to the total variation distance via\noteryan{on the left here, I think $\frac12$ can be improved to $\frac{1}{\sqrt{2}}$?} $\frac12 \delltwo{p}{q} \leq \dtv{p}{q} \leq \frac12 \sqrt{d} \cdot \delltwo{p}{q}$, using Cauchy--Schwarz.
\begin{definition}
    The \emph{Hellinger distance} between $p$ and $q$ is
    \[
        \dhell{p}{q} = \Bigl(\sum_{i=1}^d \parens*{\sqrt{p_i}-\sqrt{q_i}}^2\Bigr)^{1/2}.
    \]
    Equivalently, its square may be defined as $\dhellsq{p}{q} = 2(1-\BC{p}{q})$, where
    \[
        \BC{p}{q} = \sum_{i=1}^d \sqrt{p_i}\sqrt{q_i}
    \]
    is the \emph{Bhattacharyya coefficient} (or \emph{Hellinger affinity}) of $p$ and $q$.
\end{definition}
\noindent The Hellinger distance is also a metric; it has a maximum value of~$\sqrt{2}$, occurring when $p$ and $q$ have disjoint support.  One of its main advantages comes from the fact that the Bhattacharyya coefficient satisfies the tensorization property $\BC{p \otimes p'}{q \otimes q'} = \BC{p}{q} \cdot \BC{p'}{q'}$, where $p \otimes p'$ denotes the product distribution on $[d]^2$ arising from $p$ and $p'$.  We have the following relationship between Hellinger distance and total variation distance: $\frac12\dhellsq{p}{q} \leq \dtv{p}{q} \leq \dhell{p}{q}$.  The squared Hellinger distance is also well known to be within a small constant factor of several other popular measures of discrimination, such as the Jensen--Shannon divergence and the Le Cam (triangular) discrimination.
\begin{definition}
    The \emph{$\chi^2$-divergence} of $p$ from $q$ is
    \[
        \dchisq{p}{q} = \sum_{i=1}^d \frac{(p_i - q_i)^2}{q_i},
    \]
    which we take to be $\infty$ if $p$'s support is not a subset of $q$'s support.
\end{definition}
\noindent Unlike our previous distances, the $\chi^2$-divergence is not a metric since it is not even symmetric with respect to interchanging $p$ and $q$. (For simplicity, we may still sometimes call it a ``distance''.) One utility it has is that it bounds the squared Hellinger distance, $\dhellsq{p}{q} \leq \dchisq{p}{q}$, but can be easier to calculate: if $q$ is considered ``fixed'', then the $\chi^2$-divergence is a simple polynomial in~$p$.  Finally, we should mention that the total variation distance, the squared Hellinger distance, and the $\chi^2$-divergence are all ``$f$-divergences'', a consequence of which is that they satisfy the data processing inequality~\cite[Sec.~4]{Wu17}; i.e., none of them increases when the same stochastic operation is applied to $p$ and $q$.

\subsubsection{Distances and divergences for quantum mixed states} \label{sec:quantum-distances}
There are again many distances and divergences used for comparing two quantum states $\rho$ and~$\sigma$; see, e.g.,~\cite{GLN05}, \cite[Chap.~13]{BZ07}, \cite{Aud12} for some surveys.  All of the quantum distances we review will be \emph{unitarily invariant}, meaning that $\mathrm{D}(\rho,\sigma) = \mathrm{D}(U \rho U^\dagger, U \sigma U^\dagger)$ for all unitaries~$U$.

Many classical distances have a quantum analogue, and indeed some have \emph{several} quantum analogues.  Typically, a quantum distance between $\rho$ and $\sigma$ reduces to the analogous classical distance between $p$ and $q$ in the case that $\rho = \diag(p)$ and $\sigma = \diag(q)$ are diagonal.

In particular, for every classical $f$-divergence one can form either the ``standard quantum $f$-divergence'' (introduced by Petz) or the ``measured quantum $f$-divergence'' --- see~\cite{HM17}.  We will only consider the latter.  Given a classical  $f$-divergence $\mathrm{d}_f(\cdot, \cdot)$, one obtains the corresponding measured quantum $f$-divergence $\mathrm{D}_f(\cdot,\cdot)$ as follows:
\begin{equation}    \label{eqn:classical-to-quantum}
    \mathrm{D}_f(\rho,\sigma) = \sup_{\text{POVMs } \{E_i\}_{i=1}^N}
    \braces*{\mathrm{d}_f(p_{\rho}, p_{\sigma})}, \quad \text{where }
    p_{\xi} = (\tr(\xi E_1), \dots, \tr(\xi E_N)).
\end{equation}
In other words, the quantum divergence is defined as the maximum classical divergence that can be achieved when applying the same POVM to both states.  In this section we will encounter the measured quantum $f$-divergence corresponding to total variation distance, squared Hellinger distance, and $\chi^2$-divergence.

\begin{definition}
    The \emph{trace distance} between $\rho$ and $\sigma$ is
    \[
        \Dtr{\rho}{\sigma} = \frac12 \|\rho - \sigma\|_1.
    \]
\end{definition}
\noindent The trace distance is a metric and it has a maximum value
of~$1$, occurring when $\rho$ and $\sigma$ have orthogonal support.
Helstrom~\cite{Hel76} showed that trace distance  is the measured version of  classical total variation distance  in the sense of \Cref{eqn:classical-to-quantum}.  It therefore equals the maximum probability with which the states $\rho$ and $\sigma$ can be discriminated by some measurement.  It also follows that it satisfies the ``quantum data processing inequality''.  In other words, it can never increase when the same quantum channel (completely positive trace-preserving map) is applied to both~$\rho$ and~$\sigma$.
\begin{definition}
    The \emph{Hilbert--Schmidt distance} (or \emph{Frobenius distance})
    between $\rho$ and $\sigma$ is
    \[
        \DHS{\rho}{\sigma} = \|\rho - \sigma\|_{\textrm{HS}} =  \Bigl(\sum_{i,j=1}^d \abs*{\rho_{ij} - \sigma_{ij}}^2\Bigr)^{1/2} = \tr\parens*{(\rho - \sigma)^2}^{1/2}.
    \]
\end{definition}
\noindent This metric can be seen as analogue of the classical $\ell_2$ distance.  It is not, however, a direct analogue in the sense of \Cref{eqn:classical-to-quantum}; this is related to the fact that it does not satisfy the quantum data processing inequality.  Nevertheless, it is useful by virtue of the fact that the squared Hilbert--Schmidt distance, $\DHSsq{\rho}{\sigma} = \tr\parens*{(\rho-\sigma)^2}$, is extremely easy to compute, and that it can be related to the trace distance via Cauchy--Schwarz for matrices: $\frac12 \DHS{\rho}{\sigma} \leq \Dtr{\rho}{\sigma} \leq \frac12 \sqrt{d} \cdot \DHS{\rho}{\sigma}$.\noteryan{same question about the constant on the left}

\begin{definition}
    The \emph{Bures distance} between $\rho$ and $\sigma$ is
    \[
        \DB{\rho}{\sigma} = \bigl(2(1-\Fid{\rho}{\sigma})\bigr)^{1/2},
    \]
    where
    \[
        \Fid{\rho}{\sigma} = \|\sqrt{\rho} \sqrt{\sigma}\|_1
    \]
    is the \emph{fidelity} between $\rho$ and $\sigma$.  (The quantity $1-\Fid{\rho}{\sigma}$ is termed the \emph{infidelity}).
\end{definition}
\noindent The Bures distance is a metric and it has a maximum value of~$\sqrt{2}$, occurring when $\rho$ and $\sigma$ have orthogonal support.  The work of Fuchs and Caves~\cite{FC95} shows that the (squared) Bures distance is the measured version of classical (squared) Hellinger distance  in the sense of \Cref{eqn:classical-to-quantum}.  It follows that  $\frac12\DBsq{p}{q} \leq \Dtr{p}{q} \leq \DB{p}{q}$. It also follows that the Bures distance satisfies the quantum data processing inequality.

We more often consider the square of the Bures distance, $\DBsq{\rho}{\sigma}$, which is simply twice the infidelity.  It is also quite common to consider the squared fidelity, $\Fidsq{\rho}{\sigma}$. The squared fidelity, as shown by Uhlmann~\cite{Uhl76}, is the maximum overlap between purifications of $\rho$ and $\sigma$, where the \emph{overlap} of (mixed) quantum states $\rho'$ and $\sigma'$ is defined to be $\tr(\rho' \sigma')$.\noteryan{In particular, if one state is pure --- likely $\sigma$ in applications --- then $\tr(\rho \sigma)$ \emph{is} the fidelity}

Note that when $\rho$ and $\sigma$ are ``close'', with $\Fid{\rho}{\sigma} = 1 - \eps$, we have that $1-\Fidsq{\rho}{\sigma} \approx 2\eps$. Thus there is not much difference if one defines infidelity as $1 - \Fid{\rho}{\sigma}$ or $1-\Fidsq{\rho}{\sigma}$; these quantities are always within a factor~$2$ of each other, and also of the squared Bures distance.  Also very closely related is the \emph{quantum Hellinger affinity}, $Q_{1/2}(\rho, \sigma) = \tr(\sqrt{\rho}\sqrt{\sigma})$. It satisfies $\Fidsq{\rho}{\sigma} \leq Q_{1/2}(\rho,\sigma) \leq \Fid{\rho}{\sigma}$ and has been used to define a ``quantum Hellinger distance'' by $\mathrm{D}_{\mathrm{H}}^2(\rho,\sigma) = 2(1-Q_{1/2}(\rho,\sigma))$; see~\cite{ANSV08}.  The same bound $\Fidsq{\rho}{\sigma} \leq Q_{\text{min}}(\rho,\sigma) \leq \Fid{\rho}{\sigma}$  also holds~\cite{Aud12} for the quantity $Q_{\text{min}}\diverg{\rho}{\sigma}$ arising in the quantum Chernoff bound mentioned in \Cref{sec:asympt-state-discrim}.

\begin{definition}
    Assume $\sigma$ has full rank. The \emph{Bures $\chi^2$-divergence} of $\rho$ from $\sigma$ is
    \[
        \DBchi{\rho}{\sigma} = \tr\bigl((\rho - \sigma) \cdot \Omega_{\sigma}(\rho - \sigma)\bigr),
    \]
    where $\Omega_\sigma$ is the linear operator whose inverse is defined by $\Omega_\sigma^{-1}(A) = \frac12(\sigma A + A\sigma)$.  (There is a simple generalization to the case where $\sigma$ does not have full rank, so long as $\rho$'s support is a subset of $\sigma$'s; we will not need it, however.)  In case $\sigma = \diag(\beta_1, \dots, \beta_d)$, we obtain the following more explicit formula:
    \[
        \DBchi{\rho}{\sigma} = \sum_{i,j=1}^d \frac{2}{\beta_i + \beta_j} |\Delta_{ij}|^2, \quad \text{where } \Delta = \rho - \sigma.
    \]
\end{definition}
\noindent The Bures $\chi^2$-divergence is the measured version of the classical $\chi^2$-divergence in the sense of \Cref{eqn:classical-to-quantum}, as shown in~\cite{BC94,TV15}).  As such, it satisfies the quantum data processing inequality, and we can infer from the classical case that $\DBsq{\rho}{\sigma} \leq \DBchi{\rho}{\sigma}$. Indeed, it is known~\cite{TKR+10} that the \emph{quantum relative entropy}, $\mathrm{S}\diverg{\rho}{\sigma} = \tr(\rho(\log \rho - \log \sigma))$ is sandwiched in between: $\DBsq{\rho}{\sigma} \leq \mathrm{S}\diverg{\rho}{\sigma} \leq \DBchi{\rho}{\sigma}$. As in the classical case, such bounds are what makes the Bures $\chi^2$-divergence useful, together with its having a relatively simple formula when $\sigma$ is considered to be ``fixed''.

We close this section by commenting that, although we focus on Bures $\chi^2$-divergence, there are many generalizations of $\chi^2$-divergence to the quantum case.  For example, the ``standard quantum $f$-divergence'' version is $\tr((\rho - \sigma)^2 \sigma^{-1})$.  More generally, one may consider $\tr((\rho-\sigma) \sigma^{-\alpha} (\rho - \sigma) \sigma^{\alpha-1})$ for any $\alpha \in [0,1/2]$, and there are further possibilities. See, e.g.,~\cite{Pet96,TKR+10}, wherein it is explained that the Bures $\chi^2$-divergence takes on the smallest value among a wide family of generalizations.

\subsection{Quantum probability}\label{sec:quantum-probability}

Let $V$ be a finite-dimensional vector space over $\C$ and let $\End(V)$
denote the algebra of linear operators on $V$. An operator
$X \in \End(V)$ is self-adjoint or Hermitian if $X^\dag = X$, where
$X^\dag$ denotes the conjugate-transpose of $X$; $X$ is positive if
there exists an operator $Y \in \End(V)$ such that $X = Y^\dag Y$. For
self-adjoint operators $X, Y \in \End(V)$ we write $X \preceq Y$ provided $Y - X$ is positive. The identity operator is denoted by $\Id$, with the dimension of the underlying vector space
being inferred from the context.

\begin{definition}
  A \emph{quantum state} $\varrho$ is defined to be a positive operator
  $\varrho \in \End(V)$ with $\tr(\varrho) = 1$.
\end{definition}
\begin{definition}  \label{def:POVM}
    A \emph{positive-operator valued measurement} (POVM) $\calM$ consists of a set of positive operators that sum to the identity operator $\Id$. When a measurement $\calM = \braces*{E_1, \dotsc, E_k}$ is applied to a quantum state~$\varrho$, the \emph{outcome} is $i \in [k]$ with probability $p_i =  \tr(\varrho E_i)$.
\end{definition}
\begin{definition}  \label{def:observable}
    An \emph{observable} $\cO$ is a self-adjoint operator  $\cO \in \End(V)$.  It has a unique spectral decomposition $\cO = \lambda_1 \Pi_1 + \dotsb + \lambda_k \Pi_k$, where the $\lambda_i$'s are the distinct real eigenvalues of $\cO$, and the $\Pi_i$'s are the orthogonal projections onto the associated eigenspaces.  The projections $\{\Pi_i : i \in [k]\}$ form a POVM.

    Suppose we perform this POVM on a quantum state $\varrho \in \End(V)$ and then report the eigenvalue~$\lambda_i$ upon receiving outcome~$i$.  Then we obtain a discrete real-valued random variable~$\bx$, which takes value $\lambda_i$ with probability $\tr(\varrho \Pi_i)$ for $i = 1, \dotsc, k$.
\end{definition}

\begin{fact}
    Given an observable $\cO$ and associated real-valued random variable~$\bx$, it holds that $\E[\bx] = \tr(\varrho \cO)$.  It also holds that the observable $\cO^2$ is associated to the random variable~$\bx^2$. Thus we can compute $\Var[\bx]$ as $\tr(\varrho \cO^2) - \tr(\varrho \cO)^2$.
\end{fact}
In light of these facts, it is reasonable to define the notation $\E_{\varrho}[\cO]$ and $\Var_{\varrho}[\cO]$.  In fact, we will extend this notation to all operators, not just self-adjoint ones.

\begin{definition}
  The \emph{expectation} of operator $X \in \End(V)$ with respect to state $\varrho$ is defined
  by
  \begin{align*}
    \E_\varrho[X]
    &= \tr(\varrho X).
  \end{align*}
\end{definition}
  \noindent Since $\E_\varrho[\Id] = 1$,
  $\E_\varrho[X^\dag] = \overline{\E_\varrho[X]}$, and
  $\E_\varrho[X^\dag X] \ge 0$ for all $X \in \End(V)$, the map
  $\E_\varrho[\placeholder]$ defines a positive linear functional of
  norm $1$ on $\End(V)$. Moreover,
  $\E_{\varrho \tensor \varrho'}[\placeholder]$ satisfies the following
  tensorization property:
  $\E_{\varrho \tensor \varrho'} [\cO \tensor \cO'] = \E_\varrho[\cO]
  \cdot \E_{\varrho'}[\cO']$ for all observables
  $\cO, \cO' \in \End(V)$. The following straightforward fact says that
  $\E_\varrho[\placeholder]$ is also monotone with respect to the
  L\"{o}wner partial order.

  \begin{fact}\label{fact:expectation-monotonicity}
    If $\cO_1, \cO_2 \in \End(V)$ are observables, then
    $\cO_1 \preceq \cO_2$ if and only if
    $\E_\varrho[\cO_1] \leq \E_\varrho[\cO_2]$ for all states
    $\varrho \in \End(V)$.
  \end{fact}

\begin{definition}
  The \emph{covariance} of two operators $X_1,X_2
  \in \End(V)$ with respect to state $\varrho$ is the sesquilinear form defined by
  \begin{align*}
    \Cov_\varrho[X_1,X_2]
    &= \E_\varrho[(X_1 - \mu_1 \Id)^\dag(X_2 -
      \mu_2 \Id)] = \E_\varrho[X_1^\dagger X_2] - \mu_1^\dagger \mu_2, \quad \text{where $\mu_i = \E_\varrho[X_i]$.}
  \end{align*}
\end{definition}
\noindent Since $\Cov_\varrho[\Id,\cdot\,] = \Cov_\varrho[\placeholder, \Id] = 0$, it
  follows that $\Cov_\varrho[\placeholder,\cdot\,]$ is also
  translation-invariant in each argument; i.e.,
  $\Cov_\varrho[\cO_1 + a \Id, \cO_2 + b \Id] = \Cov_\varrho[\cO_1,
  \cO_2]$ for all $a, b \in \C$. Furthermore,
  $\Cov_{\varrho \tensor \varrho'}[\placeholder,\cdot\,]$ satisfies the
  following tensorization property,
  \begin{align*}
    \Cov_{\varrho \tensor \varrho'}[X_1 \tensor Y_1, X_2
    \tensor Y_2]
    &= \Cov_\varrho[X_1, X_2] \cdot
      \Cov_{\varrho'}[Y_1, Y_2],
  \end{align*}
  for all operators $X_1, X_2, Y_1, Y_2 \in \End(V)$. Hence,
  \begin{align}  \label{eqn:indep-analogue}
    \Cov_{\varrho \tensor \varrho'}[X_1 \tensor \Id, \Id
    \tensor X_2]
    &= 0.
  \end{align}
  When $X_1$ and $X_2$ are observables, the equality above is a quantum
  analogue of the classical fact that the covariance of independent
  random variables is zero.
\
\begin{definition}
  The \emph{variance} of operator $X \in \End(V)$ with respect to state $\varrho$ is defined by
  \begin{align*}
    \Var_\varrho[X]
    &= \Cov_\varrho[X, X].
  \end{align*}
\end{definition}

\noindent It holds that $\Var_\varrho[X] \ge 0$ for all~$X$,
  $\Var_\varrho[c \cO] = |c|^2 \Var_\varrho[\cO]$ for all $c \in \C$, and
  \begin{align*}
    \Var_\varrho \left[ \sum_{i=1}^k X_i \right]
    &= \sum_{i=1}^k \Var_\varrho[X_i] + \sum_{\substack{i,j=1 \\ i \not=
    j}}^k \Cov_\varrho[X_i, X_j]
  \end{align*}
  for all operators $X_1, \dots, X_k \in \End(V)$.

\begin{remark}
    We will ultimately only be concerned about $\E_\varrho$ and $\Var_\varrho$ as applied to observables, since our state certification algorithms will involve measuring according to observables, and then applying Chebyshev's inequality to the reported random variable~$\bx$.  Nevertheless, it will be useful in intermediate calculations  to allow $\E_\varrho$, $\Var_\varrho$, and $\Cov_\varrho$ to be applied to \emph{all} operators in $\End(V)$, even though there is not an immediate connection to classical probability when non-normal operators are involved.
\end{remark}

We end this section with a definition and lemma that will assist us in finding observables with low variance. Let $V_1$ and $V_2$ denote finite-dimensional vector spaces over $\C$
and let $\Phi : \End(V_1) \to \End(V_2)$ be a linear map.
\begin{definition}
  $\Phi : \End(V_1) \to \End(V_2)$ is \emph{positive} if $\Phi(X) \succeq 0$
  for all $X \in \End(V_1)$ with $X \succeq 0$. And, $\Phi$ is \emph{unital} if
  $\Phi(\Id) = \Id$.
\end{definition}

Suppose that $V_1 = V_2 = V$ and $\Phi$ is positive and unital. Then the
following result holds:
\begin{lemma}\label{lem:positive-unital-variance}
  If $\E_\varrho \circ \, \Phi = \E_\varrho$, then
  $\Var_\varrho[\Phi(\cO)] \le \Var_\varrho[\cO]$ for all observables
  $\cO \in \End(V)$.
\end{lemma}
\begin{proof}
  Let $\cO \in \End(V)$ be an observable. Since
  $\E_\varrho[\Phi(\cO)] = \E_\varrho[\cO]$, it suffices to show that
  $\E_\varrho[\Phi(\cO)^2] \le \E_\varrho[\cO^2]$. Since $\cO$ is
  self-adjoint and $\Phi$ is positive and unital,
  $\Phi(\cO)^2 \preceq \Phi(\cO^2)$, by the Kadison--Schwarz
  inequality~\cite{Kad52}. Hence, by
  \Cref{fact:expectation-monotonicity},
  $\E_\varrho[\Phi(\cO)^2] \le \E_\varrho[\Phi(\cO^2)]$. Since
  $\E_\varrho \circ \,\Phi = \E_\varrho$, it follows that
  $\E_\varrho[\Phi(\cO)^2] \le \E_\varrho[\cO^2]$, as needed.
\end{proof}
Thus, the class of mean-preserving positive unital maps is
variance-nonincreasing.

\begin{remark}
  Although there are other measurements that can be associated with an
  observable $\cO \in \End(V)$ apart from its spectral decomposition,
  the variance of the resulting random variables is at least
  $\Var_\varrho[\cO]$. Indeed, suppose $\calM = \set{E_1, \dotsc, E_k}$
  is a POVM and $x_1, \dotsc, x_k$ are real coefficients such that
  $\cO = x_1 E_1 + \dotsb + x_k E_k$. Let
  $\Phi : \End(\C^k) \to \End(\C^k)$ denote the map defined by
  $\Phi(A) = A_{11} E_1 + \dotsb + A_{kk} E_k$ for all
  $A \in \End(\C^k)$. Since $\calM$ is a POVM, the map $\Phi$ is
  positive and unital. Hence, by the Kadison--Schwarz
  inequality~\cite{Kad52},
  \begin{align*}
    \cO^2
    &= \Phi(\diag(x_1, \dotsc, x_k))^2
      \preceq \Phi(\diag(x_1, \dotsc, x_k)^2)
      = x_1^2 E_1 + \dotsb + x_k^2 E_k
  \end{align*}
  and the result now follows from~\Cref{fact:expectation-monotonicity}.
\end{remark}
\subsection{Representation theory}  \label{sec:rep}

Let $\gS_n$ denote the symmetric group on the alphabet $[n]$ and let
$\U(d)$ denote the group of $d \times d$ unitary matrices.

\begin{definition}
  A \emph{partition} $\lambda$ is a nonincreasing sequence of
  nonnegative integers of finite support. If
  $\lambda_1 + \lambda_2 + \dotsb = n$, then $\lambda$ is said to be a
  partition of $n$, denoted by $\lambda \vdash n$. The size of the
  support of~$\lambda$ is called the \emph{length} of the partition and
  is denoted by $\ell(\lambda)$. The \emph{power sum} symmetric
  polynomial in $d$ variables $p_\lambda(x_1, \dotsc, x_d)$ associated
  to a partition $\lambda$ of length $k$ is defined by
  $p_\lambda = p_{\lambda_1} p_{\lambda_2} \dotsm p_{\lambda_k}$, where
  $p_r(x_1, \dotsc, x_d) = x_1^r + \dotsb + x_d^r$ for all $r \ge 0$.
\end{definition}

The cycle type of a permutation $\pi \in \gS_n$ is denoted by
$\cyc(\pi)$. Sorted in nonincreasing order, $\cyc(\pi)$ is a partition
of $n$. Thus, the partitions of $n$ index the conjugacy classes of
$\gS_n$.

\begin{definition}
  Let $\RepP$ denote the unitary representation of $\gS_n$ on
  $(\C^d)^\on$ defined by
  \begin{align*}
    \RepP(\pi) \, \ket{x_1} \tensor \dotsb \tensor \ket{x_n}
    &= \ket{x_{\pi^{-1}(1)}} \tensor \dotsb \tensor \ket{x_{\pi^{-1}(n)}},
  \end{align*}
  for all $\ket{x_1}, \dotsc, \ket{x_n} \in \C^d$ and $\pi \in
  \gS_n$. Furthermore, let $\Ad_U$ be the linear map on observables
  defined by $\Ad_U(X) = (U^\on) X (U^\on)^\dag$ for all $U \in \U(d)$.
\end{definition}

\begin{definition}
  The \emph{symmetric group algebra} $\C \gS_n$ is the algebra of
  functions $f : \gS_n \to \C$. The functions $1_\pi : \gS_n \to \C$
  with $\pi \in \gS_n$ form a basis of $\C \gS_n$, where $1_\pi$ is
  defined by
  \begin{align*}
    1_\pi(\tau) =
    \begin{cases}
      1, & \pi = \tau, \\
      0, & \pi \not= \tau.
    \end{cases}
  \end{align*}
  With a slight abuse of notation, we use $\pi$ to denote the function
  $1_\pi$ and think of elements of $\C \gS_n$ as linear combinations of
  permutations $\pi \in \gS_n$. Thus, the product in $\C \gS_n$ is
  obtained by extending the product in $\gS_n$ to a bilinear map.
  $\C \gS_n$ also admits a conjugate-linear involution
  $X \mapsto X^\dag$ defined by $\pi^\dag = \pi^{-1}$ for all
  $\pi \in \gS_n$.

  The representation $\RepP$ of $\gS_n$ extends to a $*$-representation
  of the $*$-algebra $\C \gS_n$ as follows:
  \begin{align*}
    X = \sum_{\pi \in \gS_n} a_\pi \pi
    &\mapsto \sum_{\pi \in \gS_n} a_\pi \RepP(\pi) = \RepP(X).
  \end{align*}
  Since the representation $\RepP$ is unitary, it follows that
  $\RepP(X^\dag) = \RepP(X)^\dag$ for all $X \in \C \gS_n$.
\end{definition}

The center of $\C \gS_n$, denoted by $Z(\C \gS_n)$, is the set of
elements $X \in \C \gS_n$ with the property that $XY = YX$ for all
$Y \in \C \gS_n$. For all partitions $\kappa \vdash n$, let
$\cO_\kappa \in \C \gS_n$ be defined by
\begin{align*}
  \cO_\kappa
  &= \avg_{\substack{\pi \in \gS_n \\ \cyc(\pi) = \kappa}} \{\pi\}.
\end{align*}
In other words, $\cO_\kappa$ is the normalized indicator function of the
conjugacy class of permutations of cycle type~$\kappa$. The following
elementary result relates the elements $\cO_\kappa$ to the center of
$\C \gS_n$.
\begin{proposition}\label{prop:real-basis-center}
  $\{\cO_\kappa \mid \kappa \vdash n\}$ is a linear basis for
  $Z(\C \gS_n)$.
\end{proposition}
For a proof, see~\cite[Proposition 4.3.7]{GW09}. Since
$\cO_\kappa^\dag = \cO_\kappa$ for all $\kappa \vdash n$, it follows
that $\{ \cO_\kappa \mid \kappa \vdash n\}$ is also a basis for the real
vector space of self-adjoint elements of $Z(\C \gS_n)$.

\section{Efficient quantum estimators}\label{sec:efficiency}

The connection between observables and random variables presented in
\Cref{sec:quantum-probability} allows us to import notions from
classical statistics into the quantum setting. In this section, this
connection is used to define quantum estimators and introduce the notion
of statistical efficiency of a quantum estimator. These notions are used
to formulate a structure theorem for efficient quantum estimators in
situations where the statistic of interest is unitarily invariant.

As before, let $V$ be a finite-dimensional vector space over $\C$. Let
$S$ denote a set of quantum states on $V$ and let $f : S \to \R$ be a
statistic on $S$. The set $S$ serves to restrict an estimation problem
to a particular class of quantum states. $S$ will be gradually
restricted, as needed, from an arbitrary set of quantum states to a set
of multipartite quantum states of the form $\rho^\on$ or
$\rho^{\tensor m} \tensor \sigma^\on$, where $\rho$ and $\sigma$ are
quantum states on $\C^d$.

\begin{definition}
  An \emph{estimator} for $f$ is an observable $\cO \in \End(V)$ such
  that $\E_\varrho[\cO] = f(\varrho)$ for all $\varrho \in S$. An
  estimator $\cO$ is \emph{efficient} if
  $\Var_\varrho[\cO] \le \Var_\varrho[\cO']$ for all estimators
  $\cO' \in \End(V)$ for $f$.
\end{definition}

Henceforth, fix $V = (\C^d)^\on$ and let $S$ denote the set of states of
the form $\rho_1 \tensor \cdots \tensor \rho_n$, where
$\rho_1, \dots, \rho_n$ are quantum states on $\C^d$.

\begin{definition}
  A statistic $f : S \to \R$ is \emph{unitarily invariant} if
  $f \circ \Ad_U = f$ for all $U \in \U(d)$. An observable
  $\cO \in \End(V)$ is \emph{unitarily invariant} if $\Ad_U(\cO) = \cO$
  for all $U \in \U(d)$.
\end{definition}

Let $\Phi$ be the map on observables $\cO \in \End(V)$ defined by
\begin{align*}
  \Phi(\cO)
  &= \int\limits_{\U(d)} \Ad_U(\cO) \, dU,
\end{align*}
where $dU$ denotes Haar measure.  Note that $\Phi$ preserves self-adjointness and, hence, maps observables
to observables.

\begin{proposition}
  If $\cO$ is an estimator for a unitarily invariant statistic $f$, then
  $\Phi(\cO)$ is also an estimator for $f$, and
  $\Var_\varrho[\Phi(\cO)] \le \Var_\varrho[\cO]$ for all
  $\varrho \in S$.
\end{proposition}
\begin{proof}
  The map $\Phi$ is positive and unital. Since $f$ is unitarily
  invariant,
  \begin{align*}
    \E_\varrho[\Phi(\cO)]
    &= \int\limits_{\U(d)} \tr(\Ad_{U^\dag}(\varrho) \cO) \, dU
      = \int\limits_{\U(d)} f(\Ad_{U^\dag}(\varrho)) \, dU
      = \int\limits_{\U(d)} f(\varrho) \, dU
      =  \E_\varrho[\cO].
  \end{align*}
  Hence, by \Cref{lem:positive-unital-variance},
  $\Var_\varrho[\Phi(\cO)] \le \Var_\varrho[\cO]$.
\end{proof}

The following result relates the image of the map $\Phi$ to the
symmetric group algebra $\C \gS_n$ and the representation $\RepP$.  It uses the Schur--Weyl duality theorem.  For a proof, see e.g.~\cite[Proposition 2.2]{CS09}.
\begin{proposition}\label{prop:projection-rep}
  The map $\Phi$ is a projection into $\RepP(\C \gS_n)$.
\end{proposition}

Thus, if $\cO$ is an efficient estimator for a unitarily invariant
statistic $f$, then $\Phi(\cO)$ is also an efficient estimator
for~$f$. Hence, the next corollary follows immediately from
\Cref{prop:projection-rep}.

\begin{corollary}\label{cor:symmetric-estimators}
  To find an efficient estimator for a unitarily invariant statistic
  $f : S \to \R$, it suffices to consider estimators of the form
  $\RepP(X)$ with $X \in \C \gS_n$.
\end{corollary}

In light of \Cref{cor:symmetric-estimators}, we introduce the
following notation:

\begin{notation}
  Let $\E_\varrho$ be extended to a map on elements $X \in \C \gS_n$
  defined by $\E_\varrho[X] = \E_\varrho[\RepP(X)]$. Thus, $\E_\varrho$,
  $\Cov_\varrho$, and $\Var_\varrho$ are defined directly on elements of
  $\C \gS_n$ via the representation $\RepP$.
\end{notation}

If $\gamma = (i_1 \ i_2 \ \cdots \ i_\ell) \in \gS_n$, let $\tr_\gamma$
be defined by
$\tr_\gamma(\varrho) = \tr(\rho_{i_1} \rho_{i_2} \dotsm
\rho_{i_\ell})$. The following proposition establishes a formula for the
expectation $\E_\varrho[\pi]$ of a permutation $\pi \in \gS_n$ with
respect to a state $\varrho \in S$.  (Caution: $\pi$ is not in general
an observable.)

\begin{proposition}\label{prop:trace-formula}
  Let $\pi \in \gS_n$ be an arbitrary permutation. If
  $\pi = \gamma_1 \dots \gamma_k$ is a decomposition of~$\pi$ into
  disjoint cycles, including cycles of length~$1$, then
  \begin{align*}
    \E_\varrho[\pi^{-1}]
    &= \prod_{i=1}^k \tr_{\gamma_i}(\varrho).
  \end{align*}
\end{proposition}
\begin{proof}
  In light of the tensorization property of $\E_\varrho$ and the fact
  that $\varrho$ is an $n$-partite quantum state, the problem reduces
  immediately to the case when $\pi$ is an $n$-cycle. Without loss of
  generality, suppose $\pi = (1 \ 2 \ \dotso \ n)$. Thus,
  \begin{align*}
    \tr(\rho_1 \rho_2 \dotsm \rho_n)
    &= \sum_{v \in [d]^n} \bra{v_1} \rho_1 \ket{v_2}
      \dotsm \bra{v_n} \rho_n \ket{v_1} \\
    &= \sum_{v \in [d]^n} \bra{v_1} \rho_1 \ket{\pi(v)_1}
      \dotsm \bra{v_n} \rho_n \ket{\pi(v)_n} \\
    &= \sum_{v \in [d]^n} \bra{v} \varrho \ket{\pi(v)} \\
    &= \tr(\varrho \RepP(\pi^{-1})) \\
    &= \E_\varrho[\pi^{-1}]. \qedhere
  \end{align*}
\end{proof}

\begin{remark}  \label{rem:okay}
    In describing the cycle type of a permutation $\pi \in \Sym{n}$, it is common to omit mentioning $1$-cycles.  Conveniently, this would have no effect in \Cref{prop:trace-formula}, since $\tr(\rho_i) = 1$ anyway for all~$i$.
\end{remark}

\begin{definition}
  The group $\Gamma$ of \emph{permutation invariants} of the set of
  states $S$ is defined by
  \begin{align*}
    \Gamma
    &= \{\pi \in \gS_n \mid \forall \varrho \in S, \; \forall X \in \C
      \gS_n, \; \E_\varrho[\pi^{-1} X \pi] = \E_\varrho[X]\}.
  \end{align*}
  Note that the definition of $\Gamma$ depends on $S$. For all
  $X \in \C \gS_n$, let $X^\Gamma \in \C \gS_n$ be defined by
  \begin{align*}
    X^\Gamma
    &= \frac{1}{|\Gamma|} \sum_{\pi \in \Gamma} \pi^{-1} X \pi.
  \end{align*}
  Thus, $X^\Gamma \tau = \tau X^\Gamma$ for all $\tau \in \Gamma$ and
  $X \in \C \gS_n$.
\end{definition}

\begin{proposition}
  For all self-adjoint elements $\cO \in \C \gS_n$,
  $\Var_\varrho[\cO^\Gamma] \le \Var_\varrho[\cO]$.
\end{proposition}
\begin{proof}
  The map $\cO \mapsto \cO^\Gamma$ is positive and unital. Moreover,
  $\E_\varrho[\cO^\Gamma] = \E_\varrho[\cO]$ for all $\cO \in \C
  \gS_n$. Hence, by \Cref{lem:positive-unital-variance},
  $\Var_\varrho[\cO^\Gamma] \le \Var_\varrho[\cO]$.
\end{proof}

\begin{corollary}\label{cor:gamma-invariant}
  To find an efficient estimator for a unitarily invariant statistic
  $f : S \to \R$, it suffices to consider estimators of the form
  $\RepP(X)$ with $X \in \C \gS_n$ and $X \tau = \tau X$ for all
  $\tau \in \Gamma$.
\end{corollary}

The group $\Gamma$ acts on $\gS_n$ by conjugation, viz.~$\tau \in \Gamma$ acts on $\gS_n$ by $\pi \mapsto \tau^{-1} \pi
\tau$. This action partitions the group $\gS_n$ into disjoint orbits:
$\gS_n = O_1 \cup \dotsb \cup O_\ell$, where two permutations $\pi_1$
and~$\pi_2$ belong to the same orbit $O_i$ for $i \in [\ell]$ if and
only if there exists $\tau \in \Gamma$ such that
$\tau^{-1} \pi_1 \tau = \pi_2$. It is easy to see that an element
$X \in \C \gS_{m+n}$ commutes with all elements of $\Gamma$ if and only
if $X$ is constant on the orbits $O_1, \dotsc, O_\ell$ defined by
$\Gamma$. Let $\phi_i \in \C \gS_{m+n}$ denote the indicator function of
the orbit $O_i$ for $i \in [\ell]$. Thus, the set
$\braces*{\phi_1, \dotsc, \phi_\ell}$ forms a basis for the elements
$X \in \C \gS_n$ that are constant on the orbits $O_1, \dotsc,
O_\ell$. Therefore, by \Cref{cor:gamma-invariant}, it holds that:

\begin{proposition}\label{prop:orbit-estimators}
  To find an efficient estimator for a unitarily invariant statistic
  $f : S \to \R$, it suffices to consider estimators of the form
  $\RepP(X)$ with $X = a_1 \phi_1 + \dotsb + a_\ell \phi_\ell$, where
  $a_1, \dots, a_\ell \in \C$.
\end{proposition}

\subsubsection{Case: $\varrho = \rho^\on$}

Let $S$ denote the set of states of the form $\varrho = \rho^\on$, where
$\rho$ is a quantum state on $\C^d$. Let $\alpha \in \R^d$ denote the
spectrum of $\rho$ (taken in some arbitrary order).

When $\varrho$ is a state of the form $\rho^\on$, the expectation
$\E_\varrho[\pi]$ of $\pi \in \gS_n$ has a particularly simple formula:

\begin{proposition}\label{prop:trace-power-sum}
  For all $\pi \in \gS_n$ with $\cyc(\pi) = \kappa$,
  $\E_\varrho[\pi] = p_\kappa(\alpha)$.
\end{proposition}
\begin{proof}
  Let $\ell$ denote the number of disjoint cycles in the decomposition
  of $\pi$. By \Cref{prop:trace-formula},
  \begin{align*}
    \E_\varrho[\pi]
    &= \tr(\rho^{\kappa_1}) \dotsm \tr(\rho^{\kappa_\ell})
      = p_{\kappa_1}(\alpha) \dotsm p_{\kappa_\ell}(\alpha)
      = p_\kappa(\alpha). \qedhere
  \end{align*}
\end{proof}

Thus, $\E_\varrho[\pi]$ depends only on the cycle type of $\pi$. Since
the cycle types of $\pi_1 \pi_2$ and $\pi_2 \pi_1$ are equal for all
$\pi_1, \pi_2 \in \gS_n$, the following result holds:

\begin{proposition}\label{prop:expectation-commutative}
  For all $X, Y \in \C \gS_n$, $\E_\varrho[XY] = \E_\varrho[YX]$.
\end{proposition}
\begin{proof}
  For all $\pi_1, \pi_2 \in \gS_n$,
  $\cyc(\pi_1 \pi_2) = \cyc(\pi_2 \pi_1)$. Hence, by
  \Cref{prop:trace-power-sum},
  $\E_\varrho[\pi_1 \pi_2] = \E_\varrho[\pi_2 \pi_1] = p_\kappa(\alpha)$,
  where $\kappa = \cyc(\pi_1 \pi_2)$. It follows by linearity that
  $\E_\varrho[XY] = \E_\varrho[YX]$ for all $X, Y \in \C \gS_n$.
\end{proof}

Thus, we obtain the following strengthening
of~\Cref{cor:symmetric-estimators}:

\begin{proposition}\label{prop:estimator-sufficiency-center}
  To find an efficient estimator for a unitarily invariant statistic
  $f : S \to \R$, it suffices to consider estimators of the form
  $\RepP(X)$ with $X \in Z(\C \gS_n)$.
\end{proposition}
\begin{proof}
  By \Cref{prop:expectation-commutative}, $\Gamma =
  \gS_n$. The statement follows immediately from
  \Cref{cor:gamma-invariant}.
\end{proof}

The expectation $\E_\varrho[X]$ of an estimator $X \in Z(\C \gS_n)$ can
be expressed as a linear combination of $p_\kappa(\alpha)$ with
$\kappa \vdash n$ where, recall, $\alpha$ is the spectrum of
$\rho$. By~\Cref{prop:real-basis-center}, the elements
$\cO_\kappa \in \C \gS_n$ with $\kappa \vdash n$ form real a basis for
the real vector space of self-adjoint elements of $Z(\C \gS_n)$. Hence,
an estimator $X \in Z(\C \gS_n)$ can be expressed uniquely as a linear
combination of the form
\begin{align*}
  X
  &= \sum_{\kappa \vdash n} a_\kappa \cO_\kappa,
\end{align*}
where $a_\kappa \in \R$ for all $\kappa \vdash n$. Thus, by
\Cref{prop:trace-power-sum},
\begin{align*}
  \E_\varrho[X]
  &= \sum_{\kappa \vdash n} a_\kappa \E_\varrho[\cO_\kappa]
    = \sum_{\kappa \vdash n} a_\kappa p_\kappa(\alpha).
\end{align*}

Moreover, an estimator $X \in Z(\C \gS_n)$ is unique, as the following
result shows.

\begin{proposition}\label{prop:estimator-uniqueness}
  If $X_1, X_2 \in Z(\C \gS_n)$ are estimators for $f : S \to \R$, then
  $X_1 = X_2$.
\end{proposition}
\begin{proof}
  Suppose $X_1 = \sum a_\kappa \cO_\kappa$ and
  $X_2 = \sum b_\kappa \cO_\kappa$. Since $X_1$ and $X_2$ are estimators
  for $f : S \to \R$, it follows that
  \begin{align*}
    \sum_{\kappa \vdash n} a_\kappa p_\kappa(\alpha)
    &= \E_\varrho[X_1]
      = \E_\varrho[X_2]
      = \sum_{\kappa \vdash n} b_\kappa p_\kappa(\alpha).
  \end{align*}
  Thus, if $h(\alpha)$ is defined by
  \begin{align*}
    h(\alpha)
    &= \sum_{\kappa \vdash n} (a_\kappa - b_\kappa) p_\kappa(\alpha),
  \end{align*}
  then $h(\alpha) = 0$ for all $\alpha \in \R_+^d$ with
  $\norm{\alpha}_1 = 1$. Note that $h$ is a homogeneous polynomial of
  degree~$n$ in~$\alpha$. Hence, if $x \in \R_+^d$ with
  $\norm{x}_1 > 0$, then
  \begin{align*}
    h(x)
    &= h \left( \norm{x}_1 \cdot \frac{x}{\norm{x}_1} \right)
      = \norm{x}_1^n \cdot h \left( \frac{x}{\norm{x}_1} \right)
      = 0.
  \end{align*}
  Thus, $h(x) = 0$ for all $x \in \R_+^d$. Since $h$ is a polynomial, it
  follows that $h \equiv 0$. Therefore, $a_\kappa = b_\kappa$ for all
  $\kappa \vdash n$, so $X_1 = X_2$.
\end{proof}

Therefore, all observables in the center of $\C \gS_n$ are efficient
estimators:

\begin{corollary}\label{cor:efficient-estimator}
  If $X \in Z(\C \gS_n)$ is an estimator for $f : S \to \R$, then $X$ is
  efficient.
\end{corollary}
\begin{proof}
  The result follows from
  \Cref{prop:estimator-sufficiency-center} and
  \Cref{prop:estimator-uniqueness}.
\end{proof}

\begin{example}\label{ex:power-sum-estimator}
  By \Cref{cor:efficient-estimator}, $\cO_\kappa$ is an efficient
  estimator for $f(\varrho) = p_\kappa(\alpha)$. In particular, suppose
  $\kappa = (k, 1, 1, \dotsc)$, which we will denote simply as~$(k)$
  (recalling \Cref{rem:okay}).  Then $\cO_{(k)}$ is an efficient
  estimator of $f(\varrho) = \tr(\rho^k)$.
\end{example}

\subsubsection{Case: $\varrho = \rho^{\tensor m} \tensor \sigma^{\tensor n}$}

Let $S$ denote the set of states of the form
$\varrho = \rho^{\tensor m} \tensor \sigma^\on$, where $\rho$ and
$\sigma$ are quantum states on $\C^d$. Let $\alpha \in \R^d$ and
$\beta \in \R^d$ denote the spectra of $\rho$ and $\sigma$,
respectively. The group $\Gamma$ of permutation invariants of $S$ can be
described as follows:

\begin{proposition}
  $\Gamma \cong \gS_m \times \gS_n$, where $(\pi_1, \pi_2) \in \Gamma$
  embeds in $\gS_{m+n}$ in the natural way, viz.~by applying $\pi_1$ to
  $\braces*{1, \dotsc, m}$ and applying $\pi_2$ to $\braces*{m, \dotsc, m + n}$.
\end{proposition}
\begin{proof}
  Let $\Gamma$ be as in the statement of the proposition and let
  $\tau \in \Gamma$. The conjugation $\pi \mapsto \tau^{-1} \pi \tau$
  applies $\tau$ to each index in the cycle decomposition of
  $\pi$. Hence, if $\tau$ acts as in the statement of the proposition,
  then, by \Cref{prop:trace-formula},
  $\E_\varrho[\pi] = \E_\varrho[\tau^{-1} \pi \tau]$.

  Conversely, let $\tau \in \gS_{m+n}$ and suppose there exists an index
  $i \in \braces*{1, \dotsc, m}$ such that
  $\tau(i) \in \braces*{m + 1, \dotsc, m + n}$. Thus, if $\pi = (1 \ i)$,
  then $\E_\varrho[\pi] = \tr(\rho^2)$ and
  \begin{align*}
    \E_\varrho[\tau^{-1} \pi \tau] =
    \begin{cases}
      \tr(\rho \sigma), & \tau(1) \in \braces*{1, \dotsc, m}, \\
      \tr(\sigma^2), & \tau(1) \in \braces*{m + 1, \dotsc, m + n}.
    \end{cases}
  \end{align*}
  Since $\rho$ and $\sigma$ are arbitrary quantum states, it follows
  that $\tau \not\in \Gamma$.
\end{proof}

To find an efficient estimator with respect to $S$, it is sufficient, by
\Cref{prop:orbit-estimators}, to consider functions $X \in \C \gS_{m+n}$
which are constant on the orbits defined by the action of $\Gamma$ on
$\gS_{m+n}$.

\begin{notation}
  Since $\Gamma$ acts on $\gS_{m+n}$ by conjugation, the orbits of
  $\Gamma$ refine the conjugacy classes of $\gS_{m+n}$. An orbit of
  $\Gamma$ is uniquely determined by a signature consisting of a cycle
  type and a map that associates each index in the cycle type with
  either $\rho$ or $\sigma$. For instance, the signature
  $(\rho \, \sigma)$ identifies the orbit of $\Gamma$ which consists of
  all transpositions that exchange an index in $\braces*{1, \dotsc, m}$ with
  an index in $\braces*{m + 1, \dotsc, m + n}$. Note that
  $(\rho \, \sigma) = (\sigma \, \rho)$. Similarly,
  $(\rho \, \rho \, \sigma)$ denotes the set of $3$-cycles with two
  indices in $\braces*{1, \dotsc, m}$ and one index in
  $\braces*{m + 1, \dotsc, m + n}$.

  If $\mathfrak{s}$ is the signature of an orbit of $\Gamma$, let
  $\cO_{\mathfrak{s}} \in \C \gS_{m+n}$ denote the average of all
  elements in the orbit. For example, $\cO_{(\rho \, \sigma)}$ denotes
  the average of all transpositions in the $(\rho \, \sigma)$~orbit
  described above.
\end{notation}

\begin{example}
  By \Cref{prop:trace-formula}, $\cO_{(\rho \, \sigma)}$ is an estimator
  for $f(\varrho) = \tr(\rho \sigma)$.
\end{example}

Moreover, $\cO_{(\rho \, \sigma)}$ satisfies the following uniqueness
property:

\begin{proposition}
  If $X \in \C \gS_{m+n}$ is an estimator for the statistic
  $f : S \to \R$ defined by $f(\varrho) = \tr(\rho \sigma)$ and $X$ is
  of the form presented in \Cref{prop:orbit-estimators}, then
  $X = \cO_{(\rho \, \sigma)}$.
\end{proposition}
\begin{proof}
  In the case when $\rho = \sigma$, $X$ becomes an estimator for
  $\tr(\rho^2)$. Then, by \Cref{prop:trace-power-sum}, $\E_\varrho[X]$ can be
  expressed as follows:
  \begin{align*}
    \E_\varrho[X]
    &= \sum_{\kappa \vdash m+n} a_\kappa p_\kappa(\alpha),
  \end{align*}
  where $\alpha$ is the spectrum of $\rho$ and $a_\kappa \in \R$ for all
  $\kappa \vdash m + n$. Since $\E_\varrho[X] - p_2(\alpha) = 0$ for all
  $\alpha \in \R^d$ with $\norm{\alpha}_1 = 1$, it follows, as in the
  proof of \Cref{prop:estimator-uniqueness}, that $a_\kappa = 0$ for all
  $\kappa \vdash m + n$ with $\kappa \not= (2)$ and $a_{(2)} = 1$. Thus,
  in general,
  $X = a \cO_{(\rho \, \rho)} + b \cO_{(\sigma \, \sigma)} + c
  \cO_{(\rho \, \sigma)}$ with $a + b + c = 1$. Since
  $\E_\varrho[X] = \tr(\rho \sigma)$, it follows that $c = 1$ and
  $a = b = 0$.
\end{proof}

A similar argument proves the following:
\begin{proposition}
  If $X \in \C \gS_{m+n}$ is an estimator for the statistic
  $f : S \to \R$ defined by $f(\varrho) = \DHSsq{\rho}{\sigma}$ and $X$
  is of the form presented in \Cref{prop:orbit-estimators}, then
  $X = \cO_{(\rho \, \rho)} + \cO_{(\sigma \, \sigma)} - 2\cO_{(\rho \,
    \sigma)}$.
\end{proposition}

Thus, the estimators obtained for $\tr(\rho \sigma)$ and
$\DHSsq{\rho}{\sigma}$ are efficient:

\begin{corollary}\label{cor:efficient-estimator-trace}
  $\cO_{(\rho \, \sigma)}$ is an efficient estimator for
  $f(\varrho) = \tr(\rho \sigma)$.
\end{corollary}

\begin{corollary}\label{cor:efficient-estimator-hilbert-schmidt}
  $\cO_{(\rho \, \rho)} + \cO_{(\sigma \, \sigma)} - 2\cO_{(\rho \,
    \sigma)}$ is an efficient estimator for
  $f(\varrho) = \DHSsq{\rho}{\sigma}$.
\end{corollary}

\section{Hilbert--Schmidt distance and related estimation}  \label{sec:hs}

\subsection{Purity, and testing identity to the maximally mixed state}\label{sec:purity}

Let $\rho$ be a quantum state on $\C^d$, let
$\varrho = \rho^{\tensor n}$, and define $f(\varrho) = \tr(\rho^2)$. The
quantity $\tr(\rho^2)$ is called the \emph{purity} of $\rho$. One can also easily compute that the purity is the same as the squared Hilbert--Schmidt distance to the maximally mixed state, up to an additive constant: $\DHSsq{\rho}{\Id/d} = \tr(\rho^2) - 1/d$.

By \Cref{ex:power-sum-estimator}, the observable $\cO_{(2)}$ is an
efficient estimator for the statistic $f$. The following result gives an
explicit formula for the variance of $\cO_{(2)}$.

\begin{lemma}\label{lem:purity-variance}
  $\displaystyle
    \Var_\varrho[\cO_{(2)}]
    = \frac{1}{\binom{n}{2}}(1-p_2(\alpha)^2) + \frac{2(n-2)}{\binom{n}{2}}(p_3(\alpha) - p_2(\alpha)^2).$
\end{lemma}
\begin{proof}
  We may compute
  \begin{align*}
    \cO_{(2)}^2
    &= \frac{1}{\binom{n}{2}} \Id +  \frac{2(n-2)}{\binom{n}{2}}\cO_{(3)} +  \frac{\binom{n-2}{2}}{\binom{n}{2}} \cO_{(2, 2)};
  \end{align*}
  this follows from the fact that if two transpositions are chosen uniformly at random from~$\Sym{n}$, their product is the identity with probability $\frac{1}{\binom{n}{2}}$, has cycle type $(3)$ with probability $\frac{2(n-2)}{\binom{n}{2}}$, and has cycle type  $(2, 2)$ with probability  $\frac{\binom{n-2}{2}}{\binom{n}{2}}$.  Now
  \[
    \E_\varrho[\cO_{(2)}^2] =  \frac{1}{\binom{n}{2}}  +  \frac{2(n-2)}{\binom{n}{2}}p_3(\alpha) +  \frac{\binom{n-2}{2}}{\binom{n}{2}} p_{(2,2)}(\alpha) = \frac{1}{\binom{n}{2}}  +  \frac{2(n-2)}{\binom{n}{2}}p_3(\alpha) +  \parens*{1- \frac{2(n-2) + 1}{\binom{n}{2}}}p_2(\alpha)^2,
  \]
    and the lemma follows.
\end{proof}

At this point, we show how to prove our \Cref{thm:robust-HS} in the special case that $\sigma$ is known to be the maximally mixed state.  (This result was originally proven, in a slightly more opaque way, in~\cite[Theorem~4.1]{OW15}.)
\begin{proposition} (Special case of \Cref{thm:robust-HS}.)
    There is an algorithm that, given $n = O(1/\eps^2)$ copies of the state $\rho \in \C^{d \times d}$, (whp) outputs ``close'' if $\DHS{\rho}{\Id/d} \leq .99 \eps$ and outputs ``far'' if $\DHS{\rho}{\Id/d} > \eps$.
\end{proposition}
\begin{proof}
  Since $\DHSsq{\rho}{\Id/d} = \tr(\rho^2) - 1/d$, the observable
  $\cO_{(2)} - \Id/d$ is an unbiased estimator of
  $\DHSsq{\rho}{\Id/d}$. Let $\alpha \in \R^d$ denote the spectrum of
  $\rho$ and let $\Delta_i = \alpha_i - 1/d$ for all $i \in [d]$. Thus,
  \begin{align*}
    p_3(\alpha) - p_2(\alpha)^2
    &= \frac{p_2(\Delta)}{d} + p_3(\Delta) - p_2(\Delta)^2
      \le p_2(\Delta) = \DHSsq{\rho}{\Id/d}.
  \end{align*}
  Hence, by \Cref{lem:purity-variance},
  \begin{align*}
    \Var_\varrho \bracks*{\cO_{(2)} - \frac{\Id}{d}}
    &= \Var_\varrho \bracks*{\cO_{(2)}}
      \leq O\parens*{\frac{1}{n^2} + \frac{p_2(\Delta)}{n}}.
  \end{align*}
  The result now follows from \Cref{lem:chebyshev}.
\end{proof}

\subsection{Linear fidelity}

Let $\rho$ and $\sigma$ be quantum states on $\C^d$, let
$\varrho = \rho^{\tensor m} \tensor \sigma^{\tensor n}$, and define
$f(\varrho) = \tr(\rho \sigma)$. The quantity $\tr(\rho \sigma)$ is
sometimes called the \emph{overlap} or \emph{linear fidelity} between $\rho$ and $\sigma$.  By
\Cref{cor:efficient-estimator-trace}, $\cO_{(\rho \, \sigma)}$ is an
efficient estimator for the statistic $f$. The following result gives an
explicit formula for the variance of $\cO_{(\rho \, \sigma)}$.

\begin{proposition}\label{prop:variance-linear-fidelity}
  $\displaystyle
    \Var_\varrho[\cO_{(\rho \, \sigma)}]
    = \frac{1}{mn}
      + \frac{1 - m - n}{m n} \tr(\rho \sigma)^2
      + \frac1n\parens*{1-\frac{1}{m}} \tr(\rho^2 \sigma)
      + \frac1m\parens*{1-\frac{1}{n}} \tr(\rho \sigma^2).
  $
\end{proposition}
\begin{proof}
  The result follows straightforwardly from
  \begin{align*}
    \cO_{(\rho \, \sigma)}^2
    &= \frac{1}{mn} \Id
      + \parens*{1-\frac{1}{m}}\parens*{1-\frac{1}{n}} \cO_{(\rho \, \sigma) (\rho \, \sigma)}
      + \frac1n\parens*{1-\frac{1}{m}} \cO _{(\rho \, \rho \, \sigma)}
      + \frac1m\parens*{1-\frac{1}{n}} \cO_{(\rho \, \sigma \, \sigma)},
  \end{align*}
    which corresponds to the fact that product of two uniformly transpositions of type $(\rho \, \sigma)$ is: the identity probability $\frac{1}{mn}$; of type $(\rho \, \sigma) (\rho \, \sigma)$ with probability $\parens*{1-\frac{1}{m}}\parens*{1-\frac{1}{n}}$; of type  $(\rho \, \rho \, \sigma)$ with probability $\frac1n(1-\frac{1}{m})$; and of type $(\rho \, \sigma \, \sigma)$ with probability $\frac1m(1-\frac{1}{n})$.
\end{proof}

\subsection{Squared Hilbert--Schmidt distance}

Let $\rho$ and $\sigma$ be quantum states on $\C^d$, let
$\varrho = \rho^{\tensor m} \tensor \sigma^{\tensor n}$, and define
$f(\varrho) = \DHSsq{\rho}{\sigma} = \tr(\rho^2) + \tr(\sigma^2) - 2
\tr(\rho \sigma)$. By \Cref{cor:efficient-estimator-hilbert-schmidt},
$\cO_{(\rho \, \rho)} + \cO_{(\sigma \, \sigma)} - 2\cO_{(\rho \,
  \sigma)}$ is an efficient estimator for the statistic $f$.

\begin{lemma}
  $\displaystyle
    \Cov_\varrho[\cO_{(\rho \, \rho)}, \cO_{(\sigma \, \sigma)}]
    = 0.$
\end{lemma}
\begin{proof}
  Note that $\cO_{(\rho \, \rho)} = \cO_{(2)} \tensor \Id$, where
  $\cO_{(2)}$ is defined on the first $m$ components of the tensor
  product. Similarly, $\cO_{(\sigma \, \sigma)} = \Id \tensor \cO_{(2)}$,
  where $\cO_{(2)}$ is defined on the last $n$ components of the tensor
  product. Hence (recalling \Cref{eqn:indep-analogue})
  \begin{align*}
    \Cov_\varrho[\cO_{(\rho \, \rho)}, \cO_{(\sigma \, \sigma)}]
    &= \Cov_\varrho[\cO_{(2)} \tensor \Id, \Id \tensor \cO_{(2)}]
      = 0. \qedhere
  \end{align*}
\end{proof}

\begin{lemma}
  $\displaystyle
    \Cov_\varrho[\cO_{(\rho \, \rho)}, \cO_{(\rho \, \sigma)}]
    = \frac{2}{m}
      \parens*{\tr(\rho^2 \sigma) - \tr(\rho^2) \tr(\rho \sigma)}.
  $
\end{lemma}
\begin{proof}
  A permutation of type $(\rho \, \rho)(\rho \, \sigma)$ or
  $(\rho \, \rho \, \sigma)$ is uniquely determined by a product of two
  transpositions of types $(\rho \, \rho)$ and $(\rho \, \sigma)$. Hence,
  \begin{align*}
    \cO_{(\rho \, \rho)} \cO_{(\rho \, \sigma)}
    &= \frac{2}{m} \cO_{(\rho \, \rho \, \sigma)} + \parens*{1 - \frac{2}{m}}\cO_{(\rho \, \rho)(\rho \, \sigma)}.
  \end{align*}
  Therefore,
  \begin{align*}
    \Cov_\varrho[\cO_{(\rho \, \rho)}, \cO_{(\rho \, \sigma)}]
    &= \E_\varrho[\cO_{(\rho \, \rho)} \cO_{(\rho \, \sigma)}] -
      \E_\varrho[\cO_{(\rho \, \rho)}] \E_\varrho[\cO_{(\rho \, \sigma)}]
    \\
    &= \frac{2}{m} \tr(\rho^2 \sigma) + \parens*{1 - \frac{2}{m}}\tr(\rho^2)\tr(\rho \sigma) - \tr(\rho^2) \tr(\rho \sigma) \\
    &= \frac{2}{m}
      \tr(\rho^2 \sigma) - \frac{2}{m} \tr(\rho^2) \tr(\rho \sigma). \qedhere
  \end{align*}
\end{proof}

\begin{proposition}\label{prop:variance-hilbert-schmidt-estimator}
  When $m = n$,
  $\displaystyle
   \Var_\varrho[\cO_{(\rho \, \rho)} + \cO_{(\sigma \, \sigma)} - 2 \cO_{(\rho
    \, \sigma)}]
    = O \parens*{\frac{1}{n^2} + \frac{\DHSsq{\rho}{\sigma}}{n}}.
  $
\end{proposition}
\begin{proof}
  Let
  $\calV = \Var_\varrho[\cO_{(\rho \, \rho)} + \cO_{(\sigma \, \sigma)} -
  2 \cO_{(\rho \, \sigma)}]$. Since $\cO_{(\rho \, \rho)}$,
  $\cO_{(\sigma \, \sigma)} \in \C \Gamma$, $\cO_{(\rho \, \rho)}$ and
  $\cO_{(\sigma \, \sigma)}$ commute with each other and with
  $\cO_{(\rho \, \sigma)}$. Hence,
  \begin{align*}
    \calV
    &= \Var_\varrho[\cO_{(\rho \, \rho)}]
      + \Var_\varrho[\cO_{(\sigma \, \sigma)}]
      + 4 \Var_\varrho[\cO_{(\rho \, \sigma)}]
      - 4 \Cov_\varrho[\cO_{(\rho \, \rho)}, \cO_{(\rho \, \sigma)}]
      - 4 \Cov_\varrho[\cO_{(\sigma \, \sigma)}, \cO_{(\rho \, \sigma)}].
  \end{align*}
  Using prior results, we have
  \begin{align*}
    \Var_\varrho[\cO_{(\rho \, \rho)}]
    + \Var_\varrho[\cO_{(\sigma \, \sigma)}]
    %&= \frac{2}{\binom{n}{2}} + \frac{6 - 4n}{n (n - 1)} (\tr(\rho^2)^2 + \tr(\sigma^2)^2) + \frac{4n - 8}{n (n - 1)} \tr(\rho^3 + \sigma^3) \\
    &\le O \parens*{\frac{1}{n^2}} + \frac{4}{n} \parens*{\tr(\rho^3) +
      \tr(\sigma^3) - \tr(\rho^2)^2 - \tr(\sigma^2)^2},
  \end{align*}
  \begin{align*}
    4 \Var_\varrho[\cO_{(\rho \, \sigma)}]
    &= \frac{4}{n^2}
      + \frac{4 - 8n}{n^2} \tr(\rho \sigma)^2
      + \frac{4n - 4}{n^2} \tr(\rho^2 \sigma)
      + \frac{4n - 4}{n^2} \tr(\rho \sigma^2) \\
    &\le O \parens*{\frac{1}{n^2}}
      + \frac{4}{n} \parens*{\tr(\rho^2 \sigma) + \tr(\rho \sigma^2) - 2
      \tr(\rho \sigma)^2},
  \end{align*}
  and
  \begin{align*}
    - 4  \Cov_\varrho[\cO_{(\rho \, \rho)}, \cO_{(\rho \, \sigma)}]
    - 4 \Cov_\varrho[\cO_{(\sigma \, \sigma)}, \cO_{(\rho \, \sigma)}]
    &= - \frac{8}{n} \parens*{
      \tr(\rho^2 \sigma) + \tr(\rho \sigma^2)
      - \parens*{\tr(\rho^2) + \tr(\sigma^2)} \tr(\rho \sigma)}.
  \end{align*}
  Therefore,
  \begin{alignat*}{2}
    \calV
    &\le O\parens*{\frac{1}{n^2}} &&+ \frac{4}{n} \parens*{
      \tr(\rho^3) + \tr(\sigma^3) - \tr(\rho^2)^2 - \tr(\sigma^2)^2
      + \tr(\rho^2 \sigma) + \tr(\rho \sigma^2) - 2 \tr(\rho \sigma)^2
                                     } \\
    & &&- \frac{4}{n} \parens*{2 \tr(\rho^2 \sigma) + 2 \tr(\rho \sigma^2)
      - 2 \parens*{\tr(\rho^2) + \tr(\sigma^2)} \tr(\rho \sigma)} \\
    &= O \parens*{\frac{1}{n^2}} &&+ \frac{4}{n} \parens*{
      \tr(\rho^3) + \tr(\sigma^3) - \tr(\rho^2)^2 - \tr(\sigma^2)^2
      - \tr(\rho^2 \sigma) - \tr(\rho \sigma^2) - 2 \tr(\rho \sigma)^2
                                     } \\
    & &&+ \frac{4}{n} \parens*{2 \parens*{\tr(\rho^2) + \tr(\sigma^2)}
      \tr(\rho \sigma)} \\
    &= O \parens*{\frac{1}{n^2}} &&+ \frac{4}{n} \parens*{
      \tr((\rho + \sigma)(\rho - \sigma)^2)
      - (\tr(\rho^2) - \tr(\rho \sigma))^2
      - (\tr(\sigma^2) - \tr(\rho \sigma))^2} \\
    &\le O \parens*{\frac{1}{n^2}} &&+ \frac{4}{n}
      \tr((\rho + \sigma)(\rho - \sigma)^2) \\
    &\le O \parens*{\frac{1}{n^2}} &&+ \frac{4}{n}
      \norm{\rho + \sigma}_{\infty} \cdot \tr\parens*{(\rho
        - \sigma)^2} \\
    &\le O \parens*{\frac{1}{n^2}} &&+ O\parens*{\frac{1}{n}}\cdot \DHSsq{\rho}{\sigma}. \tag*{\qedhere}
  \end{alignat*}
\end{proof}

\subsection{Consequences for testing}       \label{sec:test-conseq1}
\Cref{thm:robust-HS}, which uses  $O(1/\eps^2)$-copies of unknown $\rho, \sigma$ to distinguish $\DHS{\rho}{\sigma} \leq .99 \eps$ from $\DHS{\rho}{\sigma} > \eps$, is now an immediate consequence of \Cref{lem:chebyshev} and \Cref{prop:variance-hilbert-schmidt-estimator}.\\

In the remainder of this section we give the proof of \Cref{cor:pca-trace}:
\begin{proof}
    The testing algorithm does not need to know~$\delta$, nor which of $\rho$ or $\sigma$ is $\delta$-close to rank~$k$: it simply applies the robust Hilbert--Schmidt tester \Cref{thm:robust-HS} with error parameter $c \eps/\sqrt{k}$, where $c = \frac{1}{1+1/\sqrt{2}}$.    All we need to show is an elementary fact of pure matrix analysis:  assuming $\DHS{\rho}{\sigma} \leq c\eps/\sqrt{k}$, it holds that $\Dtr{\rho}{\sigma} \leq \delta+\eps$.  Since the Hilbert--Schmidt and trace distances are symmetric we may assume that it is $\sigma$ that is close to rank~$k$; and, since these distances are unitarily invariant, we may assume that $\sigma = \diag(\beta_1, \dots, \beta_d)$, where $\beta_1 + \cdots + \beta_k \geq 1-\delta$.

    Write $\rho_A$ for the top-left $k \times k$ block of~$\rho$, write $\rho_B$ for its bottom-right $(d-k) \times (d-k)$ block, and write $\rho_{\text{off}}$ for the ``off-diagonal'' $d \times d$ matrix given by zeroing out those two blocks.  Similarly define $\sigma_A$, $\sigma_B$, and $\sigma_{\text{off}}$, so $\sigma_A = \diag(\beta_1, \dots, \beta_k)$, $\sigma_B = \diag(\beta_{k+1}, \dots, \beta_d)$, and $\sigma_C = 0$.  Now
    \begin{equation}    \label{eqn:boundme}
        2\Dtr{\rho}{\sigma} = \|\rho - \sigma\|_1 \leq \|\rho_A - \sigma_A\|_1 + \|\rho_{\text{off}} - \sigma_{\text{off}}\|_1 + \|\rho_B - \sigma_B\|_1,
    \end{equation}
    by the triangle inequality.  The matrix $\rho_A - \sigma_A$ of course has rank at most~$k$, and the matrix $\rho_{\text{off}} - \sigma_{\text{off}}$ has rank at most~$2k$ (being the sum of a $k \times (d-k)$ matrix and a $(d-k) \times k$ matrix).  Thus we use Cauchy--Schwarz to bound the first two terms on the right of~\eqref{eqn:boundme} by
    \[
        \sqrt{k} \|\rho_A - \sigma_A\|_{\textrm{HS}} + \sqrt{2k} \|\rho_{\text{off}} - \sigma_{\text{off}}\|_{\textrm{HS}} \leq \sqrt{k} \DHS{\rho}{\sigma} + \sqrt{2k}\DHS{\rho}{\sigma} \leq (1+\sqrt{2}) c \eps.
    \]
    Now if we can show
    \begin{equation}    \label{eqn:showme}
         \|\rho_B - \sigma_B\|_1 \leq 2\delta + c\eps,
    \end{equation}
    we will have bounded $2\Dtr{\rho}{\sigma}$ by $2\delta + (2+\sqrt{2})c\eps = 2\delta + 2\eps$, as needed.

    To show~\eqref{eqn:showme}, we begin with the triangle inequality:
    \[
        \|\rho_B - \sigma_B\|_1 \leq \|\rho_B\|_1 + \|\sigma_B\|_1 = \tr(\rho_B)+ \tr(\sigma_B) = (1-\tr(\rho_A)) + (1-\tr(\sigma_A)),
    \]
    where the first equality used that $\rho_B$ and $\sigma_B$ are positive, and the second used that $\rho$ and $\sigma$ have trace~$1$.  Continuing,
    \[
        (1-\tr(\rho_A)) + (1-\tr(\sigma_A)) = 2 - 2\tr(\sigma_A) + \tr(\sigma_A - \rho_A) \leq 2\delta + \|\sigma_A - \rho_A\|_1 \leq 2\delta + \sqrt{k} \|\sigma_A - \rho_A\|_{\textrm{HS}},
    \]
    where we used $1 - \tr(\sigma_A) = 1 - (\beta_1 + \cdots + \beta_k) \leq \delta$, and also Cauchy--Schwarz again.  Now~\eqref{eqn:showme} follows since $\|\sigma_A - \rho_A\|_{\textrm{HS}} \leq \DHS{\rho}{\sigma} \leq c\eps/\sqrt{k}$.
\end{proof}

\section{Quantum chi-squared estimation}    \label{sec:chisq}

\subsection{A chi-squared observable}

In this section, $\sigma$ will denote a fixed full-rank $d$-dimensional density matrix, and we will develop a natural unbiased estimator for the Bures $\chi^2$-divergence $\tr((\rho-\sigma) \cdot \Omega_\sigma (\rho - \sigma))$.  This formula suggests a natural bilinear form:
\begin{definition}
    For matrices $S,T \in \C^{d \times d}$, define the bilinear form
    \[
        \trtwo{S}{T} = \tr(S \cdot \Omega_\sigma T).
    \]
\end{definition}
This bilinear form has the following ``contraction'' property:
\begin{proposition} \label{prop:trtwo-contract}
    For any $S \in \C^{d \times d}$ it holds that
    $
        \trtwo{S}{\sigma} = \tr(S) = \trtwo{\sigma}{S}.
    $
\end{proposition}
\begin{proof}
    Both identities are direct from the definition of the $\Omega_\sigma$: the first uses $\Omega_\sigma \sigma = \Id$; the second uses $S = \tfrac12 \sigma \cdot \Omega_\sigma S + \tfrac12 \Omega_\sigma S \cdot \sigma$.
\end{proof}
It follows that
\begin{align*}
    \DBchi{\rho}{\sigma} &= \trtwo{\rho-\sigma}{\rho-\sigma} \\
    &= \trtwo{\rho}{\rho} -\trtwo{\sigma}{\rho} -\trtwo{\rho}{\sigma}+\trtwo{\sigma}{\sigma} = \trtwo{\rho}{\rho} - \tr(\rho) - \tr(\rho) + \tr(\sigma),
\end{align*}
and from this we arrive at another standard formula for the Bures $\chi^2$-divergence:
\begin{proposition}                                     \label{prop:alt-chisq-formula}
    If $\rho$ is a $d$-dimensional density matrix, then
    \[
        \DBchi{\rho}{\sigma} = \trtwo{\rho}{\rho} - 1 =\tr(\rho \cdot \Omega_\sigma \rho) - 1.
    \]
    If $\sigma = \diag(\beta_1, \dots, \beta_d)$, then $\Omega_\sigma$ acts by multiplying the $ij$-th entry by $\frac{2}{\beta_i + \beta_j} = \avg\{\beta_i,\beta_j\}^{-1} $; thus in this case,
    \[
        \DBchi{\rho}{\sigma} = \parens*{\sum_{i,j=1}^d \frac{|\rho_{ij}|^2}{\avg\{\beta_i,\beta_j\}}} - 1.
    \]
\end{proposition}
In light of the above, it is natural to define the following observable.
\begin{definition}
Assume henceforth that $\sigma = \diag(\beta_1, \dots, \beta_d)$ is diagonal.  We define the associated \emph{$\chi^2$ observable}, operating on $(\C^d)^{\otimes 2}$, as follows:\noteryan{Note that $\chiobs$ is almost the ``natural matrix representation'' of operator $\Omega_\sigma$; it's just the $(ij)(i'j')$ entry is swapped with $(ii')(jj')$.  This fact doesn't rely on $\sigma$ being diagonal.}
\begin{equation*}
    \chiobs= \sum_{i, j=1}^d \frac{\ket{ji} \bra{ij}}{\avg\{\beta_i,\beta_j\}}.
\end{equation*}
Evidently, $\E_{\rho \otimes \rho}[\chiobs] = \DBchi{\rho}{\sigma} + 1$.
\end{definition}

\begin{definition}
    Given distinct $s, t \in [n]$, we write $\chiobsij{s}{t}$ for the operator which acts on $(\C^{d})^{\otimes n}$ by applying $\chiobs$ to the $s$-th and the $t$-th tensor copies of $\C^d$ and acting as the identity on the remaining copies.  (The dependence on~$n$ in the notation is implicit.)
\end{definition}
\begin{observation} \label{obs:twisted-P}
    Observe that $\chiobsij{s}{t}$ is rather similar to the observable $\Srep((s\,t))$; however, when it swaps letters $i$ and $j$, it picks up a scalar factor of $\frac{2}{\beta_i + \beta_j}$.  Thus in comparison with
    \[
        \Srep((1\,2)) \cdot \Srep((2\,3)) = \Srep((1\,2\,3)) = \sum_{i,j,k=1}^d \ket{ijk}\bra{jki}
    \]
    we have
    \[
        \chiobsij{1}{2} \cdot \chiobsij{2}{3} = \sum_{i,j,k=1}^d \frac{\ket{ijk}\bra{jki}}{\avg\{\beta_i,\beta_j\} \cdot \avg\{\beta_i,\beta_k\}},
    \]
    the  scalar factors in the denominator arising because letters $i$ and $k$ are swapped, and then letters $i$ and $j$ are swapped.  As a consequence, rather than the matrix trilinear form mapping $(R,S,T)$ to
    \[
        \tr(\Srep((1\,2)) \cdot \Srep((2\,3)) \cdot R \otimes S \otimes T) = \tr(\Srep((1\,2\,3))\cdot R \otimes S \otimes T)
        = \sum_{i,j,k=1}^d T_{ij}S_{jk}R_{ki} = \tr(TSR)
    \]
    as in \Cref{prop:trace-formula}, we obtain the trilinear form given in the subsequent definition.
\end{observation}
\begin{definition}  \label{def:trthree}
    For matrices $R,S,T \in \C^{d \times d}$, define the trilinear form
    \[
        \trthree{R}{S}{T} = \tr(\chiobsij{1}{2} \cdot \chiobsij{2}{3} \cdot R \otimes S \otimes T) = \sum_{i,j,k=1}^d \frac{T_{ij}S_{jk}R_{ki}}{\avg\{\beta_i,\beta_j\} \cdot \avg\{\beta_i,\beta_k\}}.
    \]
\end{definition}
We again get a certain ``contraction'' property:
\begin{proposition} \label{prop:trthree-contract}
    For any $S, T \in \C^{d \times d}$ it holds that
    $
        \trthree{S}{T}{\sigma} = \trtwo{S}{T} = \trthree{\sigma}{S}{T}.
    $
\end{proposition}
\begin{proof}
    We prove the second identity, the first being similar.  When we substitute $R = \sigma$ into \Cref{def:trthree} we obtain
     \[
         \trthree{\sigma}{S}{T} = \sum_{i,j,k=1}^d  \frac{T_{ij}S_{jk}\sigma_{ki}}{\avg\{\beta_i,\beta_j\} \cdot \avg\{\beta_i,\beta_k\}}
     \]
     Since $\sigma$ is diagonal, the summands with $i \neq k$ vanish.  When $i = k$ we have $\sigma_{kk} = \beta_k$, which cancels the factor of $\avg\{\beta_i,\beta_k\}$.  We are left with
     \[
         \trthree{\sigma}{S}{T} = \sum_{j,k=1}^d \frac{T_{ij}S_{ji}}{\avg\{\beta_i,\beta_j\}},
     \]
     which is indeed $\trtwo{S}{T}$.
\end{proof}
We also observe that unlike
\[
    P((1\,2)) P((1\,2)) = \sum_{i,j=1}^d \ket{ij}\bra{ij} = \Id,
\]
we have
\begin{equation}        \label{eqn:chiobs-squared}
    \chiobs \chiobs = \sum_{i,j=1}^d \frac{\ket{ij}\bra{ij}}{\avg\{\beta_i,\beta_j\}^2},
\end{equation}
a diagonal operator, but not the identity.  Finally:

\begin{definition}
    For a given $n \geq 2$, we define the  \emph{averaged $\chi^2$ observable} on $(\C^d)^\on$ to be $\chiobsavg = \avg_{s \neq t} \{\chiobsij{s}{t}\} - \Id$, where the average is over all distinct ordered pairs $s,t \in [n]$.
\end{definition}
Evidently:
\begin{proposition}                                     \label{prop:mean-var-chi}
    $\E_{\rho^\on}[\chiobsavg] = \DBchi{\rho}{\sigma}$ and $\Var_{\rho^\on}[\chiobsavg] = \Var_{\rho^\on}[\avg_{s \neq t} \{\chiobsij{s}{t}\}]$.
\end{proposition}

\subsection{Analyzing the variance of the average chi-squared observable}

The calculation of the variance of the averaged $\chi^2$ observable, $\Var_{\rho^\on}[\avg_{s \neq t} \{\chiobsij{s}{t}\}]$, proceeds exactly as does the calculation of the variance of the purity observable in \Cref{lem:purity-variance}. We obtain:

\begin{proposition}\label{prop:chi-sq-var}
    The averaged $\chi^2$-observable has variance
    \[
        \frac{1}{\binom{n}{2}} \parens*{\tr(\chiobs^2 \rho^{\otimes 2})  - \trtwo{\rho}{\rho}^2} + \frac{2(n-2)}{\binom{n}{2}} \parens*{\trthree{\rho}{\rho}{\rho} - \trtwo{\rho}{\rho}^2}.
    \]
\end{proposition}

Introducing the shorthand $\Delta = \rho - \sigma$, we analyze the terms in \Cref{prop:chi-sq-var}.
\begin{proposition}                                     \label{prop:term1}
    $\trthree{\rho}{\rho}{\rho} - \trtwo{\rho}{\rho}^2 = \trthree{\Delta}{\Delta}{\Delta} + \trthree{\Delta}{\sigma}{\Delta} - \DBchi{\rho}{\sigma}^2$.
\end{proposition}
\begin{proof}
    This is immediate from writing $\rho = \Delta + \sigma$ and using: multilinearity of $\trthree{\cdot}{\cdot}{\cdot}$; the contraction properties \Cref{prop:trtwo-contract,prop:trthree-contract}; $\tr(\rho) = \tr(\sigma) = 1$; and, $\DBchi{\rho}{\sigma} = \trtwo{\Delta}{\Delta}$.
\end{proof}
We will ignore the subtracted $\DBchi{\rho}{\sigma}^2$ and use the following simple bound for $\trthree{\Delta}{\sigma}{\Delta}$:
\begin{proposition}                                     \label{prop:term2}
    $\trthree{\Delta}{\sigma}{\Delta} \leq 2\DBchi{\rho}{\sigma}$.
\end{proposition}
\begin{proof}
    Recalling \Cref{def:trthree} and using $\sigma = \diag(\beta_1, \dots, \beta_d)$ we get
    \[
        \trthree{\Delta}{\sigma}{\Delta} = \sum_{i,j=1}^d \frac{\Delta_{ij}\beta_j\Delta_{ji}}{\avg\{\beta_i,\beta_j\}^2} \leq 2 \sum_{i,j=1}^d \frac{|\Delta_{ij}|^2}{\avg\{\beta_i,\beta_j\}} = 2\DBchi{\rho}{\sigma},
    \]
    where the inequality used  $\frac{\beta_j}{\avg\{\beta_i,\beta_j\}}  \leq 2$.
\end{proof}
We now come to the main term in \Cref{prop:term1}:
\begin{proposition}                                     \label{prop:term3}
    Assume the smallest eigenvalue of $\sigma$ is at least $\delta$.  Then
    \[
        \trthree{\Delta}{\Delta}{\Delta} \leq \sqrt{2d/\delta}\cdot \DBchi{\rho}{\sigma}^{3/2}.
    \]
\end{proposition}
\begin{proof}
    Applying Cauchy--Schwarz to the formula in \Cref{def:trthree} gives
    \[
        \trthree{\Delta}{\Delta}{\Delta}
        \leq \sqrt{\sum_{i, j, k=1}^d \frac{|\Delta_{i j}|^2 |\Delta_{k i}|^2}{\avg\{\beta_i,\beta_j\} \cdot \avg\{\beta_i,\beta_k\}}} \cdot
        \sqrt{\sum_{i, j, k=1}^d \frac{|\Delta_{j k}|^2}{\avg\{\beta_i,\beta_j\} \cdot \avg\{\beta_i,\beta_k\}}}.
    \]
The sum inside the first square-root above is
\begin{equation*}
    \sum_{i=1}^d \left(\sum_{j=1}^d \frac{|\Delta_{i j}|^2}{\avg\{\beta_i,\beta_j\}}\right)^2
\leq \parens*{\sum_{i, j=1}^d \frac{|\Delta_{i j}|^2}{\avg\{\beta_i,\beta_j\}}}^2
= \DBchi{\rho}{\sigma}^2.
\end{equation*}
For the sum inside the second square-root above, we use the elementary fact that
\[
    \avg\{\beta_i,\beta_j\}\cdot\avg\{\beta_i,\beta_k\} \geq (\delta/2) \cdot \avg\{\beta_j,\beta_k\}
\]
when $\delta \leq \beta_i, \beta_j, \beta_k \leq 1$.  Thus this second sum is at most
\[
     d \cdot (2/\delta) \cdot \sum_{j, k} \frac{|\Delta_{j k}|^2}{\avg\{\beta_j,\beta_k\}} = (2d/\delta)\cdot \DBchi{\rho}{\sigma}.
\]
Combining the two bounds above completes the proof.
\end{proof}

We now analyze the first term in \Cref{prop:chi-sq-var}, ignoring the subtracted $\trtwo{\rho}{\rho}^2$:
\begin{proposition}                                     \label{prop:term4}
    Assume the smallest eigenvalue of $\sigma$ is at least $\delta$.  Then
    \[
        \tr(\chiobs^2 \rho^{\otimes 2}) \leq 2d^2 + (2d/\delta)\cdot\DBchi{\rho}{\sigma}.
    \]
\end{proposition}
\begin{proof}
    Using $\avg\{\beta_i,\beta_j\} \geq \sqrt{\beta_i \beta_j}$, we may bound $\tr(\chiobs^2 \rho^{\otimes 2})$ as
    \[
        \sum_{i, j=1}^d \frac{\rho_{i i} \rho_{j j}}{\avg\{\beta_i,\beta_j\}^2}
        \leq \sum_{i, j=1}^d \frac{\rho_{i i} \rho_{j j}}{\beta_i\beta_j} = \parens*{\sum_{i=1}^d\frac{\rho_{ii}}{\beta_i}}^2
      = \parens*{d+\sum_{i=1}^d \frac{\Delta_{i i}}{\beta_i}}^2 \leq 2d^2 + 2\parens*{\sum_{i=1}^d \frac{|\Delta_{i i}|}{\beta_i}}^2.
    \]
    Now using $\sqrt{\beta_i} \geq \sqrt{\delta}$ and then Cauchy--Schwarz,
    \begin{align*}
        \parens*{\sum_{i=1}^d \frac{|\Delta_{i i}|}{\beta_i}}^2 \leq (1/\delta) \cdot \parens*{\sum_{i=1}^d \frac{|\Delta_{i i}|}{\sqrt{\beta_i}}}^2 &\leq (d/\delta) \cdot \sum_{i=1}^d \frac{|\Delta_{i i}|^2}{\beta_i} \\ &\leq (d/\delta) \cdot \sum_{i,j=1}^d \frac{|\Delta_{i j}|^2}{\avg\{\beta_i,\beta_j\}} = (d/\delta) \cdot \DBchi{\rho}{\sigma}. \qedhere
    \end{align*}
\end{proof}
Combining all propositions in this section, we have established the following:
\begin{theorem}                                     \label{thm:chi-var}
    Assume the smallest eigenvalue of $\sigma$ is at least $\delta$.   Then
    \[
        \Var_{\rho^\on}[\chiobsavg] \leq \frac{1}{\binom{n}{2}} \cdot \Bigl(2d^2 + (2d/\delta)\cdot \DBchi{\rho}{\sigma}\Bigr) + \frac{2(n-2)}{\binom{n}{2}} \cdot \Bigl(\sqrt{2d/\delta} \cdot \DBchi{\rho}{\sigma}^{3/2} + 2\DBchi{\rho}{\sigma}\Bigr).
    \]
\end{theorem}

\subsection{Consequences for testing}   \label{sec:chi2-consequences}
Assume $\sigma$ is a fixed known density matrix, and we wish to estimate $\DBchi{\rho}{\sigma}$ given copies of an unknown density matrix~$\rho$.  Since we may first conjugate each copy of $\rho$ by a unitary that diagonalizes $\sigma$, we may assume without loss of generality that $\sigma$ is diagonal.  Now the average $\chi^2$ observable is an unbiased estimator for $\DBchi{\rho}{\sigma}$, and \Cref{thm:chi-var} bounds its variance provided $\sigma$'s eigenvalues are not too small.  Then from \Cref{lem:chebyshev} we immediately obtain \Cref{thm:robust-chi-sq}.

As mentioned, a corollary of \Cref{thm:robust-chi-sq} is our main \Cref{thm:maybe-main}, a robust ``far-in-fidelity vs.\ close in $\chi^2$-divergence'' tester with \emph{no} assumption about $\sigma$'s eigenvalues.  For convenience we restate and prove this theorem in the contrapositive and in terms of the squared Bures distance (which, recall, is exactly half the infidelity and is upper-bounded by the $\chi^2$-divergence):
\begin{corollary}[Equivalent to \Cref{thm:maybe-main}]                                     \label{cor:robust-bures}
    Fix a $d$-dimensional mixed state~$\sigma$.  Then there is an algorithm that, given $n = O(d/\eps)$ copies of~$\rho$, (whp) outputs ``close'' if $\DBchi{\rho}{\sigma} \leq .49 \eps$ and outputs ``far'' if $\DBsq{\rho}{\sigma} > .5\eps$.
\end{corollary}
\begin{proof}
    Let $\Phi_\eta$ denote the depolarizing channel, which maps a state $\nu \in \C^{d \times d}$ to the state $\Phi_\eta(\nu) = (1-\eta)\nu + \eta \Id/d$.  Define $\rho' = \Phi_{c\eps}(\rho)$ and $\sigma' = \Phi_{c\eps}(\sigma)$, where $c > 0$ is a small absolute constant to be chosen later.

    If $\DBchi{\rho}{\sigma} \leq .49\eps$ then $\DBchi{\rho'}{\sigma'} \leq .49\eps$ by the quantum data processing inequality.  On the other hand, in case $\DBsq{\rho}{\sigma} > .5 \eps$,
    \begin{equation}\label{eq:bures-is-a-metric-yknow}
        \sqrt{.5\epsilon} < \DB{\rho}{\sigma} \leq \DB{\rho}{\rho'} + \DB{\rho'}{\sigma'} + \DB{\sigma'}{\sigma}
    \end{equation}
    by the triangle inequality. We can bound the first of these terms by
    \begin{equation*}
        \DBsq{\rho}{\rho'}  \leq 2 \Dtr{\rho}{\rho'} = \norm*{\rho - \rho'}_1 = \norm*{c \eps \rho + c\eps \Id/d}_1 \leq 2c\eps,
    \end{equation*}
    where at the end we used the triangle inequality and $\|\rho\|_1 = \|\Id/d\|_1 = 1$. A similar argument shows that $\DBsq{\sigma}{\sigma'} \leq 2c\eps$; i.e., $\DB{\sigma}{\sigma'} \leq \sqrt{2c\eps}$.  Now taking $c$ sufficiently small, \eqref{eq:bures-is-a-metric-yknow}~implies $\DB{\rho'}{\sigma'} > \sqrt{.495\eps}$ and hence $\DBchi{\rho'}{\sigma'} \geq \DBsq{\rho'}{\sigma'} \geq .495\eps$.

    In summary, if $\DBchi{\rho}{\sigma} \leq .49\eps$ then $\DBchi{\rho'}{\sigma'} \leq .49\eps$, if $\DBsq{\rho}{\sigma} > .5\eps$ then $\DBchi{\rho'}{\sigma'} > .495\eps$, and all the eigenvalues of $\sigma'$ are at least $c \eps/d$.  Thus we can obtain the desired tester by first applying the depolarizing channel $\Phi_{c\eps}$ to the $n$ copies of $\rho$, producing $n$ copies of $\rho'$, and then using the tester from \Cref{thm:robust-chi-sq} with $\sigma'$ in place of $\sigma$ and $.5\eps$ in place of $\eps^2$.
\end{proof}

We can also use this corollary to test if an unknown state is diagonal:
\begin{theorem}\label{thm:diag-test}
    Given $n = O(d/\eps)$ copies of a $d$-dimensional mixed state~$\rho$, one can distinguish (whp) the case that $\rho$ is diagonal (in the standard basis) from the case that $\rho$ has infidelity more than~$\eps$ with every diagonal state.
\end{theorem}

\begin{proof}
    Let $p = (\rho_{11}, \dots, \rho_{dd})$ denote the diagonal of~$\rho$, a probability distribution.  We can obtain a sample from $p$ given a copy of~$\rho$ simply by measuring $\rho$ in the standard basis.  As mentioned near \Cref{eqn:classical-chisq}, $O(d/\eps)$ samples suffice produce an estimate $\wh{p}$ of~$p$ that satisfies $\dchisq{p}{\widehat{p}_{\mathrm{diag}}} \leq .49\eps$ (whp).  The tester now applies \Cref{cor:robust-bures} with $\sigma = \diag(\wh{p})$, using another $O(d/\eps)$ samples.  If $\rho$ is diagonal, then $\DBchi{\rho}{\sigma} = \dchisq{p}{\wh{p}} \leq .49\eps$ and the tester outputs ``close'' (whp).  If~$\rho$ has infidelity more than~$\eps$ with every diagonal state, then in particular $1-\Fid{\rho}{\sigma} > \eps$; i.e., $\DBsq{\rho}{\sigma} > .5 \eps$, and the tester outputs ``far'' (whp).
\end{proof}

\section{Implementing the observables}\label{sec:implementing}

In this section, we give efficient algorithms implementing some of our observables.
In \Cref{sec:purity-implement}, we implement the purity observable from \Cref{sec:purity},
in \Cref{sec:implement-hilbert-schmidt}, we implement the Hilbert--Schmidt observable from \Cref{sec:hs},
and in \Cref{sec:alternative}, we implement a different, though related, observable for the Hilbert--Schmidt distance.

Our main tool is \emph{Schur--Weyl duality} from the representation theory of the symmetric and general linear groups. We assume familiarity with representation theory; see~\cite{GW09}.

\begin{notation}
Given a partition $\lambda \vdash n$,
we write $\SYT{\lambda}$ for the set of \emph{standard Young tableaus} of shape~$\lambda$
and $\SSYT{\lambda}{d}$ for the set of \emph{semistandard Young tableaus} of shape~$\lambda$ and alphabet~$[d]$.
\end{notation}

\begin{notation}\label{def:schur-weyl}
Recall the representations $\RepP(\pi)$ and $\RepQ(M)$ of the symmetric and general linear groups, respectively, which act on the vector space $(\C^{d})^{\otimes n}$.
Because these two commute with each other, $\RepP(\pi)\cdot\RepQ(M)$ is a representation of the product group $\Sym{n}\times \GL(d)$.
\emph{Schur--Weyl duality} describes how $(\C^{d})^{\otimes n}$ decomposes under this group action:
\begin{equation}\label{eq:schur-weyl}
(\C^{d})^{\otimes n} \cong \bigoplus_{\substack{\lambda \vdash n\\ \ell(\lambda) \leq d}} \Specht{\lambda} \otimes \Weyl{\lambda}{d},
\end{equation}
where $\Specht{\lambda}$ and $\Weyl{\lambda}{d}$ are the irreducible representations of the symmetric and general linear groups, respectively, corresponding to~$\lambda$.
We write $\sirrep{\lambda}(\pi)$
for the matrix
associated with the symmetric group irreducible representation
at the permutation $\pi \in \Sym{n}$.
\end{notation}

\subsection{Implementing the purity observable}\label{sec:purity-implement}

In this section, we describe how to compute the~$\cO_{(2)}$ observable for estimating the purity,
which we used in \Cref{sec:purity} to test whether a state is maximally mixed.
We begin by deriving the eigendecomposition for \emph{all}~$\cO_{\mu}$ observables.

\begin{notation}
Given a partition $\lambda \vdash n$, we write $\Pi_\lambda$ for the projector onto the $\lambda$-isotypic subspace in \Cref{eq:schur-weyl}.
If $\ell(\lambda) > d$, then $\Pi_\lambda$ is just the all-zeros matrix.
\end{notation}

\begin{proposition}\label{prop:eigendecomposition}
For any partition $\mu \vdash k \leq n$,
\begin{equation*}
\cO_{\mu}=
\sum_{\lambda} \frac{\chi_\lambda(\mu\cup 1^{n-k})}{\dim(\lambda)} \cdot \Pi_\lambda.
\end{equation*}
\end{proposition}
\begin{proof}
By definition of $\cO_{\mu}$,
\begin{equation*}
\cO_{\mu}
= \avg_{\substack{\pi \in \Sym{n}\\\cyc(\pi) = \mu}}\{\RepP(\pi)\}
\cong \bigoplus_{\substack{\lambda \vdash n\\ \ell(\lambda) \leq d}} \avg_{\substack{\pi \in \Sym{n}\\\cyc(\pi) = \mu}}\{\sirrep{\lambda}(\pi)\} \otimes I_{\dim(\Weyl{\lambda}{d})}
= \bigoplus_{\substack{\lambda \vdash n\\ \ell(\lambda) \leq d}}
	\frac{\chi_\lambda(\mu\cup 1^{n-k})}{\dim(\lambda)} \left(I_{\dim(\lambda)} \otimes I_{\dim(\Weyl{\lambda}{d})}\right),
\end{equation*}
where the last step is by Schur's lemma and the fact that $\tr(\sirrep{\lambda}(\pi)) = \chi(\mu \cup 1^{n-k})$ if $\cyc(\pi) = \mu$.
The right-hand side equals the expression in the proposition, as the $I_{\dim(\lambda)} \otimes I_{\dim(\Weyl{\lambda}{d})}$ term
just projects into the $\lambda$-isotypic subspace.
\end{proof}

Hence, to implement the $\cO_{\mu}$ observable,
we  measure according to the $\Pi_\lambda$ projectors and output $\chi_\lambda(\mu \cup 1^{n-k})/\dim(\lambda)$.
As we will see, this can be done efficiently for $\mu = (2)$.

\begin{definition}\label{def:schur-projector}
\emph{Weak Schur sampling} refers to performing the projective measurement $\{\Pi_\lambda\}_\lambda$ on the space $(\C^d)^{\otimes n}$.
It can be implemented in time $\poly(n,d)$; see, for example, \cite{MdW16}.
\end{definition}

\begin{definition}\label{def:purityestimator}
Given a partition $\lambda \vdash n$, the \emph{second moment estimator} is defined as
\begin{equation*}
\purityestimator{\lambda}:= \frac{\chi_{\lambda}(2 \cup 1^{n-2})}{\dim(\lambda)}.
\end{equation*}
In general, computing the characters of the symmetric group is $\#\PTIME$-hard~\cite{Hep94}
(in fact, even deciding whether a character is nonzero is $\NP$-hard~\cite{PP17}).
However, Frobenius~\cite{Fro00} gives an explicit formula for the character ratio~$\purityestimator{\lambda}$
(see Ingram~\cite{Ing50} for a simple proof of this formula).
The following equivalent expression is found, for example, in~\cite{IO02}:
\begin{equation}\label{eq:second-moment-polynomial}
\purityestimator{\lambda} = \frac{1}{n(n-1)} \sum_{i=1}^d \left((\lambda_i - i + \tfrac{1}{2})^2 - (-i + \tfrac{1}{2})^2\right).
\end{equation}
\end{definition}

As a result, because weak Schur sampling
and computing $\purityestimator{\lambda}$
are both efficient operations, we can conclude with the following theorem.

\begin{theorem}
The $\cO_{(2)}$ observable can be computed in time $\poly(n,d)$.
\end{theorem}

\noindent
We note that this is the same algorithm as~\cite{OW15} used for testing whether a state is maximally mixed,
and it was previously used by~\cite{CHW07} to distinguish the maximally mixed state from states which are maximally mixed on a subspace of dimension~$d/2$.
For a more intuitive view of this algorithm,
suppose we perform weak Schur sampling on $\rho^{\otimes n}$,
where~$\rho$ is a density matrix with sorted eigenvalues $\alpha = (\alpha_1, \ldots, \alpha_d)$.
A long line of work~\cite{ARS88, KW01, HM02, CM06, OW16, OW17} has shown that the random measurement outcome~$\blambda$,
when rescaled as $\blambda/n := (\blambda_1/n, \ldots, \blambda_d/n)$,
is a good approximation to~$\alpha$.
To estimate the purity~$p_2(\alpha)$ of~$\alpha$, then,
it is natural to output a statistic close to~$p_2(\blambda/n)$,
and $\purityestimator{\blambda}$ is the apparent appropriate statistic.

\begin{remark}
The $\cO_{\mu}$ observables are related to the \emph{central characters},
defined for any $\lambda \vdash n$ and $\mu \vdash k$ as
\begin{equation*}
        p^\#_\mu(\lambda) = \begin{cases}
                            n^{\downarrow k}\cdot \frac{\chi_\lambda(\mu \cup 1^{n-k})}{\dim(\lambda)} & \text{if }n \geq k,\\
                            0 & \text{if }n < k,
                                 \end{cases}
\end{equation*}
where $n^{\downarrow k} = n (n-1) \cdots (n-k+1)$.
For~$\mu$ fixed, these are polynomials which are \emph{shifted-symmetric} in the $\lambda_i$'s, in the sense of~\cite{OO98b},
of which \Cref{eq:second-moment-polynomial} is a special case; see~\cite{IO02} for a particularly thorough treatment of these polynomials.
Our rule for multiplying the $\cO_{\mu}$'s can be viewed as deriving from the multiplication rule for $p^\#_\mu$ polynomials due to~\cite{IK01}.
\end{remark}

\subsection{Implementing the Hilbert--Schmidt observable}\label{sec:implement-hilbert-schmidt}

In this section, we describe how to compute the~$\cO_{(\rho \rho)} + \cO_{(\sigma \sigma)} - 2 \cO_{(\rho \sigma)}$
observable for estimating the squared Hilbert--Schmidt distance between~$\rho$ and~$\sigma$,
which we used in~\Cref{sec:hs} to test whether~$\rho$ and~$\sigma$ are equal.
There, we considered the general case of states $\rho^{\otimes m} \otimes \sigma^{\otimes n}$, for $m$ possibly not equal to~$n$.
For simplicity, we will restrict ourselves to the case $m = n$, though our argument easily extends to the general case.
In this section, and this section only, we will write the observable $\cO_{(2)} \in \C \gS_k$, for a given integer $k$, as  $\cO_{(2)}^k$, so as to make the~$k$ dependence explicit.
Given this, we can rewrite our Hilbert--Schmidt observable in the following manner.

\begin{proposition}\label{prop:weird-coeffs}
\begin{equation*}
\cO_{(\rho \rho)} + \cO_{(\sigma \sigma)} - 2 \cO_{(\rho \sigma)}
= \left(\frac{2n-1}{n}\right) \cdot \cO_{(\rho \rho)}
+ \left(\frac{2n-1}{n}\right) \cdot \cO_{(\sigma \sigma)}
 - \left(\frac{4n-2}{n}\right) \cdot \cO_{(2)}^{2n}.
\end{equation*}
\end{proposition}
\begin{proof}
The observable $\cO_{(2)}^{2n}$ decomposes as
\begin{equation*}
\cO_{(2)}^{2n}
= \left(\frac{n-1}{4n-2}\right) \cdot \cO_{(\rho \rho)} + \left(\frac{n-1}{4n-2}\right) \cdot \cO_{(\sigma \sigma)} + \left(\frac{2n}{4n-2}\right) \cO_{(\rho \sigma)},
\end{equation*}
where the weights correspond to the probabilities that a random $2$-cycle from $\gS_n$ either falls in the first half of $[2n]$, the second half, or falls in both halves.
The proposition follows by substitution.
\end{proof}

If we note that $\cO_{(\rho \rho)} = \cO_{(2)}^{n} \otimes \Id$ and $\cO_{(\sigma \sigma)} = \Id \otimes \cO_{(2)}^{n}$,
where $I$ is the identity matrix acting on $(\C^d)^{\otimes n}$,
then
by \Cref{prop:eigendecomposition} and \Cref{def:purityestimator},
we can rewrite the first two terms in \Cref{prop:weird-coeffs} as
\begin{equation}\label{eq:simultaneously-diagonalizable}
\left(\frac{2n-1}{n}\right) \cdot \cO_{(\rho \rho)}
+ \left(\frac{2n-1}{n}\right) \cdot \cO_{(\sigma \sigma)}
= \left(\frac{2n-1}{n}\right) \sum_{\lambda, \mu \vdash n} (\purityestimator{\lambda} + \purityestimator{\mu}) \cdot \Pi_\lambda \otimes \Pi_\mu.
\end{equation}
We can also rewrite the third term in \Cref{prop:weird-coeffs} as
\begin{equation}\label{eq:thing-its-simultaneously-diagonalizable-with}
\left(\frac{4n-2}{n}\right) \cdot \cO_{(2)}^{2n} = \left(\frac{4n-2}{n}\right)\sum_{\nu \vdash 2n} \purityestimator{\nu} \cdot \Pi_\nu.
\end{equation}
We note that $\cO_{(2)}^{2n}$ commutes with $\cO_{(\rho \rho)}$  and $\cO_{(\sigma \sigma)}$.
This is because both of these latter matrices are elements of $\C \gS_{2n}$,
and by \Cref{prop:real-basis-center} we know that the center of $\C \gS_{2n}$ contains $\cO_{(2)}^{2n}$.
By linearity, \eqref{eq:simultaneously-diagonalizable} therefore commutes with~\eqref{eq:thing-its-simultaneously-diagonalizable-with},
and as a result these two matrices are simultaneously diagonalizable,
with joint eigenspaces corresponding to the projectors $(\Pi_\lambda \otimes \Pi_\mu) \Pi_\nu = \Pi_\nu(\Pi_\lambda \otimes \Pi_\mu)$.
Applying \Cref{prop:weird-coeffs}, we have that
\begin{equation*}
\cO_{(\rho \rho)} + \cO_{(\sigma \sigma)} - 2 \cO_{(\rho \sigma)}
= \sum_{\substack{\lambda, \mu \vdash n\\\nu \vdash 2n}}\left(\left(\frac{2n-1}{n}\right) (\purityestimator{\lambda} + \purityestimator{\mu}) - \left(\frac{4n-2}{n}\right) \purityestimator{\nu}\right) \cdot \Pi_\nu (\Pi_\lambda \otimes \Pi_\mu).
\end{equation*}
This equation immediately gives us our algorithm.
\begin{theorem}
Given $\rho^{\otimes n}$ and $\sigma^{\otimes n}$, the Hilbert--Schmidt observable can be computed as follows:
\begin{itemize}
\item Perform weak Schur sampling on $\rho^{\otimes n}$ and $\sigma^{\otimes n}$, receiving~$\bmu, \bnu \vdash n$, respectively.
\item Perform weak Schur sampling on all $2n$ qudits, receiving~$\blambda \vdash 2n$.
\item Output
\begin{equation*}
\left(\frac{2n-1}{n}\right) \cdot \purityestimator{\bmu} + \left(\frac{2n-1}{n}\right) \cdot \purityestimator{\bnu} - \left(\frac{4n-2}{n}\right) \cdot \purityestimator{\blambda}.
\end{equation*}
\end{itemize}
As noted in \Cref{def:schur-projector}, the Hilbert--Schmidt observable can therefore be computed in time $\mathrm{poly}(n,d)$.
\end{theorem}

\subsection{An alternative Hilbert--Schmidt observable}\label{sec:alternative}

In the case when the input is $\varrho = \rho^{\otimes n}$ and one already knows~$\sigma$,
one can estimate the squared Hilbert--Schmidt distance between~$\rho$ and~$\sigma$
by outputting~$m$ copies of~$\sigma$ and measuring the observable from \Cref{sec:implement-hilbert-schmidt}.
In this section, we record an alternative observable which performs the same task without first preparing copies of~$\sigma$.

\begin{definition}
For a word $w \in [d]^n$,
its \emph{type} is given by $\tau = (\tau_1, \ldots, \tau_d)$,
where $\tau_i$ is the number of $i$'s in~$w$, for each $i \in [d]$.
Write $\types{n}{d}$ for the set of types corresponding to words in~$[d]^n$;
then $(\tau_1, \ldots, \tau_d) \in \types{n}{d}$ if and only if each $\tau_i$ is a nonnegative integer and $\tau_1 + \cdots + \tau_d = n$.
The \emph{$\tau$-subspace} is the span of all vectors~$\ket{x}$ of type~$\tau$;
we write $\Pi_\tau$ for the corresponding projector.
\end{definition}

\begin{definition}
Given $\sigma = \mathrm{diag}(\beta)$, for $\beta = (\beta_1, \ldots, \beta_d)$,
we define the \emph{inner-product observable}
\begin{equation*}
\mathsf{IP} = \sum_{\tau \in \types{n}{d}} \frac{\langle \beta, \tau\rangle}{n} \cdot \Pi_\tau.
\end{equation*}
Its name refers to its expectation, $\E_{\varrho}[\mathsf{IP}]= \tr(\rho \sigma)$.
The alternative Hilbert--Schmidt observable is
\begin{equation*}
\cO_{(2)} + \tr(\sigma^2) \cdot \Id - 2 \cdot \mathsf{IP}.
\end{equation*}
By \Cref{ex:power-sum-estimator}, this has expectation $\E_{\varrho}[\cO_{(2)} + \tr(\sigma^2) \cdot \Id - 2 \cdot \mathsf{IP}] = \tr(\rho^2) + \tr(\sigma^2) - 2 \tr(\rho\sigma) =  \DHSsq{\rho}{\sigma}$.
\end{definition}

We see that this observable is an unbiased estimator for the squared Hilbert-Schmidt distance.
Its variance can be analyzed using the same techniques as for our other observables.
Doing so yields a bound that matches the variance of the normal Hilbert-Schmidt observable.

\begin{theorem}
This observable has variance
\begin{equation*}
\Var_\varrho[\cO_{(2)} + \tr(\sigma^2) \cdot \Id - 2 \cdot \mathsf{IP}]  = O\parens*{\frac{1}{n^2} + \frac{\DHSsq{\rho}{\sigma}}{n}}.
\end{equation*}
Applying \Cref{lem:chebyshev}, we rederive \Cref{thm:robust-HS} for the case of known~$\sigma$:
 $n= O(1/\epsilon^2)$ copies of~$\rho$ are sufficient to distinguish $\DHS{\rho}{\sigma} \leq .99\epsilon$ from $\DHS{\rho}{\sigma} \geq \epsilon$.
\end{theorem}

To implement this observable, we will need to find a common orthogonal basis for both $\cO_{(2)}$ and $\mathsf{IP}$.
This is provided by the following definition.

\begin{definition}\label{def:GZ}
Fix a Young diagram $\lambda \vdash n$.
The \emph{Young--Yamanouchi basis} of $\Specht{\lambda}$ has a vector $\ket{S}$ for each standard Young tableau $S \in \SYT{\lambda}$.
Similarly, the \emph{Gelfand--Tsetlin basis} of $\Weyl{\lambda}{d}$ has a  vector $\ket{T}$ for each semistandard Young tableau $T \in \SSYT{\lambda}{d}$.
By \Cref{def:schur-weyl}, the vectors $\ket{\lambda} \otimes\ket{S}\otimes \ket{T}$, ranging over all $\lambda\vdash n$, $S \in \SYT{\lambda}$, and $T\in\SSYT{\lambda}{d}$, therefore form a basis for the space $(\C^d)^{\otimes n}$.
Furthermore, this basis has the following property:
\begin{center}
Write $\tau = (\tau_1, \ldots, \tau_d) \in \types{n}{d}$ for the \emph{type} of $\ket{T}$,  where $\tau_i$ is the number of $i$'s in~$T$,
\\ for each $i \in [d]$.
Then $\ket{\lambda} \otimes\ket{S}\otimes \ket{T}$ is contained in the $\tau$-subspace of $(\C^d)^{\otimes n}$.
\end{center}
The unitary transformation which maps the standard basis into this basis is known as the \emph{Schur transform},
and by the work of~\cite{BCH05,Har05} it can be computed in time $\mathrm{poly}(n,d)$.
\end{definition}

Consider the \emph{$(\lambda, \tau)$-subspace} of $(\C^d)^{\otimes n}$,
i.e., the subspace spanned by those vectors of the form $\ket{\lambda}\otimes\ket{S}\otimes\ket{T}$, where $T$ has type~$\tau$.
Then by \Cref{def:GZ}, this is a subspace of both $\Pi_\lambda$ and $\Pi_\tau$ and is therefore simultaneously an eigenspace
for the~$\cO_{(2)}$ and~$\IP$ observables.
As a result, writing $\Pi_{\lambda, \tau}$ for the projector onto this subspace,
we may decompose our observable as
\begin{equation*}
\cO_{(2)} + \tr(\sigma^2) \cdot \Id - 2 \cdot \mathsf{IP}
= \sum_{\lambda, \tau} \left(\purityestimator{\lambda} + \tr(\sigma^2) - 2 \frac{\langle \beta, \tau\rangle}{n}\right) \cdot \Pi_{\lambda, \tau}.
\end{equation*}
As we have seen, we can perform the $\Pi_{\lambda, \tau}$ measurement using the Schur transform.
We can also compute it by performing the $\{\Pi_{\lambda}\}_\lambda$ measurement (i.e.,\ weak Schur sampling) followed by the $\{\Pi_{\tau}\}_\tau$ measurement,
using the fact that $\Pi_{\lambda, \tau} = \Pi_\lambda \Pi_\tau = \Pi_\tau \Pi_\lambda$.
In conclusion, we derive the following algorithm.

\begin{theorem}
Given $\rho^{\otimes n}$, the alternative Hilbert--Schmidt observable can be computed as follows:
\begin{itemize}
\item Measure $\rho^{\otimes n}$ in the Gelfand--Tsetlin basis, receiving a semistandard tableau~$\bT$ of shape~$\blambda$ and type~$\btau$.
\item Output
$
\purityestimator{\blambda} + \tr(\sigma^2) - \langle \beta, \btau\rangle/n.
$
\end{itemize}
Alternatively, we can receive~$\blambda$ and~$\btau$
by first performing weak Schur sampling and then performing the $\{\Pi_\tau\}_{\tau}$ projective measurement.
By \Cref{def:schur-projector} and \Cref{def:GZ},
both of these algorithms compute the alternative Hilbert-Schmidt observable in time $\mathrm{poly}(n,d)$.
\end{theorem}

\bibliographystyle{alpha}
\bibliography{quantum}
\end{document}